\documentclass[10pt,journal,compsoc]{IEEEtran}
\usepackage[cmex10]{amsmath}
\usepackage{amsfonts}
\usepackage{amssymb}
\usepackage{bbm}
\usepackage{verbatim} 
\usepackage{graphics,psfrag,epsfig}
\usepackage{epstopdf} 
\ifCLASSOPTIONcompsoc
  \usepackage[nocompress]{cite}
\else
  \usepackage{cite}
\fi
\usepackage{cases}
\usepackage{adjustbox}
\usepackage[linesnumbered,ruled,vlined,algo2e]{algorithm2e}
\usepackage{empheq}
\usepackage{tcolorbox}

\usepackage{amsthm} 
\usepackage{hyperref}

\usepackage{color}
\definecolor{red}{rgb}{1,0.2,0.2}
\definecolor{green}{rgb}{0.2,1,0.5}
\definecolor{blue}{rgb}{0,0,0}
\definecolor{lightblue}{rgb}{0.3,0.5,1}

\newcommand{\blue}[1]{\textcolor{blue}{#1}}

\newcommand{\diag}{\mathtt{diag}}

\newcommand{\0}{\mathbf{0}}
\newcommand{\1}{\mathbf{1}}

\newcommand{\cb}{\mathbf{c}}

\newcommand{\hh}{\mathbf{h}}

\newcommand{\bsigma}{\boldsymbol{\sigma}}
\newcommand{\uu}{\mathbf{u}}
\newcommand{\vv}{\mathbf{v}}

\newcommand{\x}{{\mathbf{x}}}

\newcommand{\RN}{\mathbb{R}^N}
\newcommand{\Rn}{\mathbb{R}^n}
\newcommand{\Rnn}{\mathbb{R}^{n\times n}}

\newcommand{\cc}{\mathbf{c}}

\newcommand{\gb}{\mathbf{g}} 
\newcommand{\hb}{\mathbf{h}}

\newcommand{\pb}{\mathbf{p}}

\newcommand{\bs}{\mathbf{s}}
\newcommand{\sbold}{\mathbf{s}}
\newcommand{\T}{^{\mathsf{T}}}
\newcommand{\nT}{^{-\mathsf{T}}}
\newcommand{\tb}{\mathbf{t}}
\newcommand{\ub}{\mathbf{u}}

\newcommand{\xx}{{\mathbf{x}}}
\newcommand{\yb}{\mathbf{y}}
\newcommand{\zb}{\mathbf{z}}

\newcommand{\nnb}{\nonumber}

\newtheorem{thm}{Theorem}%

\theoremstyle{remark}
\newtheorem{rem}{Remark}

\newcommand {\cA}{{\mathcal{A}}}

\newcommand {\cP}{{\mathcal{P}}}

\newcommand {\cS}{{\mathcal{S}}}

\newcommand {\bp} {{\bf p}}

\newcommand {\bx} {{\bf x}}
\newcommand {\by} {{\bf y}}

\newcommand {\blambda} {\boldsymbol{\lambda}}
\newcommand {\bmu} {\boldsymbol{\mu}}

\newcommand {\btheta} {\boldsymbol{\theta}}

\newcommand{\avg}{{\rm avg}}

\newcommand {\N} {{\mathbb{N}}}
\newcommand {\R} {\mathbb{R}}
\newcommand {\C} {{\rm I\kern-5pt C}}

\newtheorem{lemma}{Lemma}

\newcommand{\beqa}{\begin{eqnarray}}
\newcommand{\eeqa}{\end{eqnarray}}
\newcommand{\beqan}{\begin{eqnarray*}}
	\newcommand{\eeqan}{\end{eqnarray*}}
\newcommand{\beq}{\begin{equation}}
\newcommand{\eeq}{\end{equation}}

\newcommand{\bfl}{\begin{flushleft}}
	\newcommand{\efl}{\end{flushleft}}
\newcommand{\bgnv}{\begin{verbatim}}
	\newcommand{\ednv}{\end{verbatim}}

\newcommand{\myb}{\hspace{-0.1in}}

\newcommand{\myeq}{& \hspace{-0.1in} = & \hspace{-0.1in}}

\newcommand{\lb}{\nonumber \\}
\newcommand{\myarr}{\begin{array}{lll}}

	\newcommand{\mygeq}{& \myb \geq & \myb}

	\newcommand{\bitem}{\begin{itemize}}
		\newcommand{\eitem}{\end{itemize}}
	\newcommand{\benum}{\begin{enumerate}}
		\newcommand{\eenum}{\end{enumerate}}
	
	\newcommand{\lnorm}{\left| \left|}
	\newcommand{\rnorm}{\right| \right|}

	\newcommand{\myskip}{\\ \vspace{-0.1in}}

\newtcolorbox{mymathbox}[1][]{colback=white, sharp corners, #1}
	
\begin{document}
	
\title{Optimal Cybersecurity Investments in 
Large Networks Using
SIS Model: Algorithm Design}

\author{Van Sy Mai, Richard J. La, Abdella Battou\thanks{V.-S. Mai and A. Battou
are with the National Institute of Standards and
Technology (NIST), Gaithersburg, MD 20899, USA. 
Email: \{vansy.mai, abdella.battou\}@nist.gov. 
R.J. La
is with NIST and the University of Maryland, 
College Park, MD 20742, USA.

Any mention of commercial products in this paper is for information only; it does not imply recommendation or endorsement by NIST.}}
	

%

\IEEEtitleabstractindextext{%
\begin{abstract}
We study the problem of minimizing the (time) 
average security costs in large networks/systems 
comprising many interdependent subsystems, 
where the state evolution is captured by a susceptible-infected-susceptible (SIS) model. 
The security costs reflect 
security investments, economic losses
and recovery costs from infections and failures 
following successful attacks. We show that the
resulting optimization problem is nonconvex
and propose a suite of algorithms -- two based on
a convex relaxation, and the other two for finding
a local minimizer, based on a reduced gradient 
method and sequential convex programming. 
Also, we provide a sufficient condition under which 
the convex relaxations are exact and, hence, their
solution coincides with that of 
the original problem. Numerical 
results are provided to validate our analytical
results and to demonstrate the effectiveness of
the proposed algorithms.
\end{abstract}

}

\maketitle

\IEEEdisplaynontitleabstractindextext

%
\IEEEpeerreviewmaketitle

\IEEEraisesectionheading{\section{Introduction}\label{sec:Introduction}}

%
%
%
%


\IEEEPARstart{T}{oday}, many modern engineered systems, including 
information and communication networks and 
power systems, comprise many interdependent 
systems. For uninterrupted delivery of their services, 
the comprising systems must work together 
and oftentimes support each other.
Unfortunately, this interdependence among
comprising systems also introduces a source of 
vulnerability in that
it is possible for a local failure or
infection of a system by malware to spread 
to other systems, potentially 
compromising the integrity of the
overall system. Analogously, in social networks,
contagious diseases often spread from infected 
individuals to other vulnerable individuals 
through contacts or physical proximity.

From this viewpoint, it is clear that the 
underlying networks that govern the interdependence
among systems have a large impact on dynamics of 
the spread of failures or malware infections. 
Similarly, the topology and contact frequencies 
among individuals in social networks
significantly influence the manner in which 
diseases spread in societies. 
Thus, any sound
investments in security of complex 
systems or the control of epidemics
should take into account the interdependence
in the systems and social contacts 
in order to maximize potential benefits
from the investments. 

In our model, attacks targeting the systems
arrive according to some (stochastic) process. 
Successful attacks on the systems can also spread
from infected systems to other systems via
aforementioned dependence among the systems. 
The system operator decides appropriate security 
investments to fend off the attacks, which
in turn determine their breach probability, 
i.e., the probability that they fall victim to 
attacks and become infected. 

Our goal is to minimize
the (time) average costs of a system 
operator managing a large system 
comprising many systems, such as large 
enterprise intranets. The overall 
costs in our model account for both 
security investments and recovery/repair 
costs ensuing infections or failures, 
which we call {\em infection costs}
in the paper.
To this end, we first consider a scenario 
where malicious actors launch {\em external} 
or {\em primary} attacks. When a primary
attack on a system is
successful, the infected system can spread 
it to other systems, which we call {\em 
secondary} attacks by the infected systems, 
to distinguish them from primary 
attacks. When primary attacks do not stop, 
it is in general not possible to achieve 
an infection-free state at steady state. 

In the second case, we assume that there 
are no primary attacks and examine the
steady state, starting with an initial 
state where some systems are infected. 
The goal of studying this scenario is to
get additional insights into scenarios
where the primary attacks occur
infrequently. It turns out that, 
even in the absence of primary attacks, 
infections may persist due to secondary
attacks and the system may not be able
to attain the infection-free steady 
state, and we can compute an upper bound
on the optimal value, which is also
tight under some condition, more easily. 

We formulate the problem of determining the 
optimal security investments that minimize 
the average costs as an optimization problem. 
Unfortunately, this optimization problem is
nonconvex and 
cannot be solved easily. In order to 
gauge the quality of a feasible solution, 
we obtain both a lower bound and an upper
bound on the optimal value of our problem. 
A lower bound can be acquired using one
of two different convex relaxations of the 
original problem we propose. For the 
convex relaxations, we also derive a sufficient 
condition under which the solution of the 
convex relaxation solves the original 
nonconvex optimization problem
(Lemma~\ref{corExactRelaxation}). 
An upper bound on the optimal value can 
be obtained using an algorithm that finds
a local minimizer. Here, we propose two
methods -- 
a reduced gradient method (RGM)
and sequential convex programming (SCP), both 
of which produce a local minimizer. 
Together, our approach offers 
a bound on the optimality gap. 

Numerical studies show that the computational
requirements for the proposed methods are 
light to modest even for large systems, except
for one method, which requires the calculation 
of an inverse matrix. They suggest
that, in almost all cases that we considered, the 
gap between the lower bound on the optimal 
value and the cost achieved by our solutions
is small; in fact, in most cases, 
the gap is less than
2-3 percent with the gap being less than 0.3
percent in many cases. 
In addition, when the infection 
costs are large, which are likely true in many 
practical scenarios, the sufficient condition 
for the convex relaxations to be exact
holds, and we obtain optimal points by solving
the convex relaxations. 
Finally, the RGM is 
computationally most efficient (with the
computational time being less than two
seconds in all considered cases and less
than 0.1 seconds in most cases) and 
the quality of solutions is on par with
that of other methods. This suggests that
the RGM may offer
a good practical solution for our problem.

\subsection{Related Literature}
    \label{subsec:Related}

\blue{Given the importance of cybersecurity,
robustness of complex systems, and control
of epidemics, there is already a large 
body of literature that examines how
to optimize the (security) investments in 
complex systems~\cite{Jiang11, LelargeBolot-NetEcon08}, 
the mitigation of disease or infection
spread~\cite{Cohen-PRL03, Gourdin-DRCN11, 
Nowzari-TCNS17}, 
feasibility and case studies of cyber 
insurance using pre-screening or differentiated
pricing based on the security investments of
the insured~\cite{Pal19, Khalili18}, and
designing good attack models and 
effective mitigating defense 
against attacks~\cite{Miehling18, Manadhata11, Sheyner03}.
In view of the volume of existing literature, here we 
summarize only a small set of studies most closely 
related to our study.}

In \cite{Khalili-NetCon19, La-ToN16, 
LelargeBolot-NetEcon08, Hota-TCNS18}, 
the authors adopted
a game theoretic formulation to
study the problem of security investments 
with distributed agents or autonomous 
systems that do not coordinate
their efforts. \blue{The problem 
we study in this paper is complementary, but 
is very different from
those studied in the aforementioned studies: 
in our setting, we assume that the 
system is managed by a {\em single}
operator and is interested in minimizing the
average (security) costs over time by 
determining (nearly) optimal security investments. 
Our study is applicable to, for example, the
problem of finding suitable security investments
in large enterprise intranets supporting 
common business processes or supervisory control 
and data acquisition systems comprising many
subsystems.}

In another line of research, which is most
closely related to our study, researchers
investigated optimal strategies using
vaccines/immunization (prevention) 
\cite{Cohen-PRL03,Preciado-CDC13}, antidotes 
or curing rates (recovery) \cite{Borgs-RSA10,  
mai2018distributed, Ottaviano-JCN18} 
or a combination of
both preventive and recovery measures
\cite{Nowzari-TCNS17, Preciado-TCNS14}. 
For example, \cite{Preciado-CDC13}
studies the problem of partial vaccination
via investments at each individual 
to reduce the infection rates, with the
aim of maximizing the exponential decay 
rate to control the spread of an epidemic. 
Similarly, \cite{mai2018distributed} 
examines the problem of determining the
optimal curing rates for distributed
agents under different formulations. 
In particular, the last formulation
of the problem \cite{mai2018distributed}, 
for which only partial result is obtained, 
is closely related to a special case of
our formulation studied in 
Section~\ref{sec:Case_lambda0}.

Key differences between existing studies, 
including those listed above, and ours
can be summarized as follows: 
first, unlike previous studies that
focus on either the expected costs from
single or cascading failures/infections 
\cite{LelargeBolot-NetEcon08, Kunreuther03, 
La-ToN16} or the exponential decay rate to 
the disease-free state as a key performance 
metric, we aim to minimize
the (time) average costs of a system 
operator, while accounting for both 
security investments and infection costs, 
with both primary and secondary attacks, 
by modeling time-varying states of systems
due to the transmissions of failures/infections.
In the presence of 
primary attacks, it is in general not 
possible to achieve the infection-free
steady state and, thus, the exponential
decay rate is no longer a suitable 
performance metric for our study. 
Second, unlike some studies that assume
that the expected costs/risks seen by
systems are convex functions of security 
investments
(e.g., \cite{Jiang11}), the expected risks 
are derived from the steady state 
equilibrium of a differential system that 
describes system states and depends on 
security investments. As it will be clear, 
the lack of a closed-form expression for the
steady state equilibrium complicates 
considerably the
analysis and algorithm design. 

Preliminary results of this paper were reported  
in \cite{mai2020optimal}. In this paper, we 
extend the findings of \cite{mai2020optimal}
in several significant directions. 
First, we offer an alternative convex
relaxation of the original problem. As
we demonstrate, this new relaxation 
technique based on exponential cones, is
more efficient and avoids the key issues
of the approach based on M-matrix theory 
presented in \cite{mai2020optimal}. 
Second, we present another computationally
efficient approach to finding a suboptimal 
solution using SCP together with 
our M-matrix theory-based
approach, which provides an upper 
bound on the optimal value of the original 
nonconvex optimization problem. We compare
this approach to one based on the RGM
and show that the quality 
of solutions from these two methods is 
comparable, but the RGM
holds a slight computational edge. 
Finally, we study a special case with no
primary attacks, which is related to the
epidemic control problem studied in 
\cite{Mieghem-TON09, mai2018distributed, 
Ottaviano-JCN18}. 
Although this can be viewed as 
a limit case of our problem formulation
as the rates of primary attacks go to zero, 
our approaches to obtaining upper and lower 
bounds of the optimal value require 
significant modifications for the reason
explained in Section~\ref{sec:Case_lambda0}. 
Moreover, we derive sufficient conditions for 
optimality, which can be verified relatively
easily.

The rest of the paper is organized as follows:
Section~\ref{sec:Prelim} explains the notation 
and terminology we adopt. Section~\ref{sec:Formulation}
describes the setup and the problem formulation, 
including the optimization problem. 
Section~\ref{sec:LowerBound} discusses two
different convex relaxations of the original problem, followed by two methods
for finding local minimizers in Section
\ref{sec:UpperBound}. We discuss a special case 
with no primary attacks in Section
\ref{sec:Case_lambda0}. Numerical results are 
provided in Section~\ref{sec:Numerical}, \blue{
followed by a discussion on how our formulation and
results can be extended 
in Section~\ref{sec:Discussion}.} We 
conclude in Section~\ref{sec:Conclusion}.

\section{Preliminaries}
	    \label{sec:Prelim}
		    
\subsection{Notation and Terminology}
    \label{subsec:Notation}

Let $\mathbb{R}$ and $\mathbb{R}_+$ denote the 
set of real numbers and nonnegative real 
numbers, respectively. Given a set $\mathbb{A}$, 
we denote the closure, interior, and boundary of
$\mathbb{A}$ by cl$(\mathbb{A})$, int$(\mathbb{A})$, 
and $\partial \mathbb{A}$, respectively. 

For a matrix $A=[a_{i,j}]$, let $a_{i,j}$ denote 
its $(i,j)$ element, $A\T$ its transpose, 
$\rho(A)$ its spectral radius, 
and $\underline{\sigma}(A)$ and $\bar{\sigma}(A)$ the 
smallest and largest real parts of its eigenvalues. 
For two matrices $A$ and $B$, we write $A \!\ge\! B$ if 
$A\!-\!B$ is a nonnegative matrix. 
%
We use boldface letters to denote vectors, e.g., 
$\x \!=\! [x_1,..., x_n]\T$ and 
$\1 \!=\! [1,...,1]\T$. 
For any two vectors $\xx$ and $\yb$ of the same 
dimension, $\xx \circ \yb$ and $\frac{\xx}{\yb}$ are 
their element-wise product and division, respectively. 
For $\xx \in \mathbb{R}^n$, $\diag(\xx) \in  
\mathbb{R}^{n \times n}$ denotes the diagonal matrix 
with diagonal elements $x_1,\ldots,x_n$.

A directed graph $\mathcal{G} \!=\! (\mathcal{V}, 
\mathcal{E})$ consists of a set of nodes $\mathcal{V}$, 
and a set of directed edges $\mathcal{E}  \subseteq  
\mathcal{V} \times \mathcal{V}$. A directed path is a 
sequence of edges in the form $\big( (i_1, i_2), 
(i_2, i_3),..., (i_{k-1}, i_k) \big)$. The graph 
$\mathcal{G}$ is strongly connected if there is a 
directed path from each node to any other node.
\\ \vspace{-0.1in}

\subsection{M-Matrix Theory} 
	\label{subsec:M-matrix}

A matrix $A \in \R^{n \times n}$ is an 
M-matrix if it can be expressed in the form $A = 
sI - B$, where $B \in \Rnn_+$ and 
$s \ge \rho(B)$. The set of (nonsingular) $n\times n$ 
M-matrices is denoted by ($\mathbb{M}^{n\times n}_+$) 
$\mathbb{M}^{n\times n}$. Note that this definition 
implies that the off-diagonal elements of $A$ are 
nonpositive and the diagonal elements are nonnegative; 
any matrix satisfying these conditions is called a 
Z-matrix. 
We shall make use of the following results on the 
properties of a nonsingular M-matrix \cite{plemmons1977m}. 

\begin{lemma}\label{lem_M_Matrix}
Let $A\in\mathbb{R}^{n\times n}$ be a Z-matrix. Then, 
$A \in \mathbb{M}^{n\times n}_+$ 
if and only if one of the following conditions holds: 

\begin{itemize}

\item[(a)] $A + D$ is nonsingular for every diagonal $D\in
\Rnn_+$. 

\item[(b)] $A$ is inverse-positive, i.e., $\exists A^{-1} 
\in \R^{n \times n}_+$. 

\item[(c)] $A$ is monotone, i.e., $A\xx \ge 0 \Rightarrow 
\xx \ge 0, \forall \xx \in \Rn$.

\item[(d)] Every regular splitting of $A$ is convergent, 
i.e., if $A = M-N$ with $M^{-1}, N \in \R^{n \times n}_+$, 
then $\rho(M^{-1}N) <1$.

\item[(e)] $A$ is positive stable, i.e., 
$\underline{\sigma}(A) >0$. 

\item[(f)] $\exists \xx>0$ with $A\xx \ge 0$ such that
if $[A\xx]_{i_0} = 0$, then 
$\exists i_1,\!..., i_r$ with 
$[A\xx]_{i_r} \!\!>\! 0$ and 
$a_{i_k, i_{k\!+\!1}} \!\!\neq\! 0, \forall k\!\in\! [0,r\!-\!1]$.


\item[(g)] $\exists \xx>\0$ with $A\xx > \0$. 

\end{itemize}
\end{lemma}

The next result is a direct consequence of  
\cite[Thm.~2]{mai2018optimizing}.

\begin{lemma}\label{lemConvexityInvMMatrix}
Let $A \in \mathbb{M}^{n\times n}$ be irreducible. Then 
\begin{itemize}

\item[(i)] $\diag(\zb) + A \in \mathbb{M}^{n\times n}_+$ 
for every $\zb \in \Rn_+\setminus\{\0\}$.

\item[(ii)] $\big[\big( \diag(\zb) + A \big)^{-1}\big]_{i,j}$  
is a convex and decreasing function in $\zb \in 
\mathbb{R}^n_+$ for all $1 \le i,j \le n$.

\end{itemize}%
\end{lemma}

\section{Model and Formulation}
    \label{sec:Formulation}
		
Consider a large system consisting of $N$ 
systems, and denote the set of comprising 
systems by ${\cal A} := \{1, 2, \ldots, N\}$. 
The security of the systems is interdependent 
in that the failure or infection 
of a system can cause that of other 
systems.\footnote{Throughout the 
paper, we use the words `failure' and 
`infection' interchangeably, in order to 
indicate \mbox{that a system fell victim to an
attack.}}  
As stated before, we study the problem 
of determining security investments for 
hardening each system in order to defend	
the systems against attacks 
in large systems, in which the
comprising systems depend on each other 
for their function. The goal of the system 
operator is to minimize the average 
aggregate costs for all systems (per unit time), 
which account for both security investments 
and any economic losses from 	
failures/infections of systems. 
\\ \vspace{-0.1in}

\subsection{Setup}
	\label{subsec:Setup}    
		
We assume that each system experiences 
primary attacks from malicious actors. 
Primary attacks on system $i \in {\cal A}$ 
occur in accordance with a Poisson 
process with rate $\lambda_i \in \R_+$.
When a system experiences 
an attack, it suffers an infection 
and subsequent economic losses 
with some probability, 
called {\em breach probability}. 

This breach probability depends on the 
security investment on the 
system: let $s_i \in 
\R_+$ be the security investment on 
system $i$ (e.g., investments in monitoring 
and diagnostic tools). The breach probability of
system $i$ is determined by some function 
$q_i: \R_+ \to (0, 1]$. In other words, 
when the operator invests $s_i$ on system
$i$, its breach probability is equal 
to $q_i(s_i)$. We assume that $q_i$ is 
decreasing, strictly convex and continuously
differentiable for all $i \in {\cal A}$.
It has been shown~\cite{Bary-WEIS12} that, 
under some conditions, 
the breach probability is decreasing and 
log-convex.

When system $i$ falls victim to an attack
and becomes infected, the operator incurs 
costs $c^r_i$ per unit time for recovery 
(e.g., inspection and repair of servers). 
Recovery times are modeled
using independent and identically distributed
(i.i.d.) exponential random variables with
parameter $\delta_i > 0$. Besides 
recovery costs, the infection of system 
$i$ may cause economic 
losses if, for example, some servers in 
system $i$ have to be taken 
offline for inspection and repair and are
inaccessible during the period to other 
systems that depend on the servers. 
To model this, we assume that the
infection of system $i$ introduces economic
losses of $c^e_i$ per unit time. 

Besides primary attacks, systems also 
experience secondary attacks from other
infected
systems. For example, this can model the 
spread of virus/malware or 
failures in complex systems.
The rate at which the infection of system $i$
causes that of another system $j$ is denoted
by $\beta_{i,j} \in \R_+$. When $\beta_{i, j}
> 0$, we say that system $i$ supports system
$j$ or system $j$ depends on 
system $i$. Let $B = [b_{i,j} : i, j
\in \mathcal{A}]$ be an $N \times N$
matrix that describes the infection rates
among systems, 
where the element $b_{i,j}$ is equal to 
$\beta_{j,i}$. We adopt the convention
$\beta_{i,i} = 0$ for all $i \in {\cal A}$. 

Define a directed graph ${\cal G}
= ({\cal A}, {\cal E})$, where a directed
edge from system $i$ to system $j$, denoted
by $(i, j)$, 
belongs to the
edge set ${\cal E}$ if and only if 
$\beta_{i,j} > 0$. 
We assume that matrix $B$ is 
\underline{irreducible}. 
Note that this is equivalent to 
assuming that the graph ${\cal G}$
is strongly connected.

\subsection{Model}
\label{subsec:Model}
   
We adopt the well-known
susceptible-infected-susceptible (SIS) model 
to capture the evolution of system
state. Let $p_i(t)$ be the probability that 
system $i$ is at the `infected' state (I) 
at time $t \in \R_+$. 
We approximate the dynamics of ${\bf p}(t)
:= (p_i(t) : i \in {\cal A})$, 
$t \in \R_+$, using the following (Markov)
differential equations, which are derived in 
\cite{Mieghem-TON09} and are based on mean
field approximation. This model is 
also similar to those employed in 
\cite{Gourdin-DRCN11, mai2018distributed,
Nowzari-TCNS17, Ottaviano-JCN18, 
Preciado-CDC13, Preciado-TCNS14}: 
for fixed security investments, 
${\bf s} = (s_i : \ i \in {\cal A}) \in 
\R_+^N$, 
\beqa 
\dot{p}_i(t)
= (1 - p_i(t)) q_i(s_i) \Big(
\lambda_i + \sum_{j \in {\cal A}} 
\beta_{j, i} p_j(t) \Big) 
- \delta_i p_i(t).
\label{eq:pdot}
\eeqa

In practice, the breach probability 
$q_i$ can be a complicated
function of the security investment. 
Here, in order to make progress, we assume 
that the breach probability functions 
can be approximated (in the regime of
interest) using a function
of the form $q_i(s) = (1 + \kappa_i
s)^{-1}$ for all $i \in {\cal A}$. 
The parameter $\kappa_i > 0$ models
how quickly the breach probability decreases
with security investment for system $i$. 
The assumed function 
satisfies log-convexity shown in 
\cite{Bary-WEIS12}. 

Define $\alpha_i := \kappa_i \delta_i$, 
$i \in \mathcal{A}$, and $\boldsymbol{\alpha}
:= (\alpha_i : \ i \in \mathcal{A})$. The 
following theorem tells us that, for a fixed
security investment vector ${\bf s} 
:= (s_i : \ i \in {\cal A}) \in \R_+^N$, 
there is a unique equilibrium of the 
differential system described by \eqref{eq:pdot}. 
 
\begin{thm} \label{thmEquilibrium}
Suppose $\boldsymbol{\lambda} \gneq \0$, 
$\boldsymbol{\delta} > \0$ and $\sbold \ge \0$ 
are fixed. If the network is strongly connected, 
i.e., $B$ is irreducible, there exists a 
unique equilibrium $\pb^* \in (0,1)^N$ of
\eqref{eq:pdot}. Moreover, starting with 
any $\pb_0$ satisfying $ \pb^* \!\leq\! 
\pb_0 \!\leq\! \1$, 
the iteration
\begin{align}
\pb_{k+1}  
= \frac{\boldsymbol{\lambda} + B\pb_k}
    {\boldsymbol{\lambda} + B\pb_k 
    + \boldsymbol{\alpha}\circ \bs 
    + \boldsymbol{\delta}}, \ k \in \N, 
\label{eqFixedPointEquilibrium}
\end{align}
converges linearly to $\pb^*$ with some rate 
$\rho_0 < 1-\min_{i \in {\cal A}} p_i^*$.
\end{thm}
\begin{proof}
Please see Appendix~\ref{appen:Theorem1}. 
\end{proof}

Note that the unique equilibrium of the
differential system given by \eqref{eq:pdot}
specifies the probability that each system will 
be infected at steady state. For this reason, we take 
the average cost of the system, denoted by 
$C_{\avg}(\sbold)$, to be 
\beqa
\hspace{-0.2in}
C_{\avg}(\sbold)
& \myb := & \myb w(\sbold)
+ \textstyle \sum_{i \in {\cal A}} c_i
p_i^*(\sbold)
= w(\sbold) + \cb\T \pb^*(\sbold), 
    \label{eq:Cavg}
\eeqa
where $c_i := c^r_i + c^e_i$, ${\bf c} = 
(c_i : i \in {\cal A})$, and $w(\sbold)$
quantifies the security investment costs
(per unit time), 
e.g., \ $w(\sbold) = \sum_{i \in 
{\cal A}} s_i$. 
We assume that $w$ is continuous, (weakly) 
convex and strictly increasing, and 
refer to $\cb$ simply as the infection 
costs (instead of infection costs per
unit time). 

A major difficulty in minimizing the average cost in 
\eqref{eq:Cavg} as an objective function is that the 
equilibrium $\pb^*(\bs)$ does not have a closed-form 
expression. \blue{As
a result, we cannot simply substitute
a closed-form expression for the equilibrium
$\pb^*(\bs)$ in \eqref{eq:Cavg}
and minimize the average cost with $\bs$ 
as the optimization variables.}
For this reason, we formulate the problem 
of determining optimal security investments that 
minimize the average cost $C_{\avg}({\bf s})$
as follows:
\begin{tcolorbox}[colback=white]
\vspace{-0.21in}
\begin{subequations}
\begin{align}
\hspace{-0.68in}
(\mathrm{P})\qquad \quad  \min_{\bs \ge \0, \pb\ge \0}&\quad  
f(\bs, \pb) := w(\bs) + \cb\T\pb 
	\label{eqObjectiveFun}\\
\mathsf{s.t.}&\quad \gb (\bs, \pb) = \0 
	\label{eqSteadyState}
	\vspace{-0.65in}
\end{align}
\end{subequations}
\vspace{-0.30in}
\end{tcolorbox}
\noindent where $\gb (\bs, \pb) = (g_i(\bs, \pb) : 
 i \in \mathcal{A})$, and 
\begin{eqnarray*}
g_i(\bs, \pb) 
\myeq (1-p_i) \Big(\lambda_i + \sum_{j \in \mathcal{A}} 
	\beta_{j, i} p_j \Big) 
    - (\alpha_i s_i + \delta_{i})p_i, \ i \in {\cal A}.
\end{eqnarray*}
Recall that, for given $\sbold \in 
\mathbb{R}^N_+$, only the unique 
equilibrium $\pb^* \in (0,1)^N$ in Theorem 
\ref{thmEquilibrium} satisfies the 
constraint in \eqref{eqSteadyState}. 
\blue{Clearly, the solution to problem (P) will also 
shed light
on which systems are more critical from the security 
perspective and, hence, should be protected.}
\blue{In the problem (P), we do not explicitly model any 
total budget constraint on security investments
for simplicity of exposition.
However, we will revisit the issue of
constraints on security investments, such as 
a total budget constraint, and discuss how it 
affects our main results in Section
\ref{sec:ConstraintDiscussion}.}

This problem (P) is nonconvex due to the nonconvexity 
of the equality constraint functions in 
\eqref{eqSteadyState}. In particular, $g_i$ contains 
both quadratic or bilinear terms $p_i p_j$ and $p_i 
s_i$. In the following sections, we develop four
complementary algorithms for finding good-quality
solutions to the nonconvex problem:
the first two approaches are based on a convex 
relaxation using different techniques, and provide 
a lower bound on the optimal value of the problem 
(P). The last two are designed to find a local 
minimizer of the problem (P), hence provide 
an upper bound on the optimal value, and are based 
on the RGM and SCP.

\section{Lower Bounds via Convex Relaxations} 
    \label{sec:LowerBound}

In this section, we discuss how we can relax
the original problem (P) and construct two
different convex formulations, which can be
used to obtain (a) a lower bound on the optimal
value and (b) a feasible solution to 
(P) using optimal points of the relaxed problems. 
Furthermore, we provide a sufficient condition for 
the relaxed problems to be exact, i.e., their 
optimal point is also an optimal point of 
the nonconvex problem (P). 
The first approach is based on M-matrix theory
and the preliminary results were reported in 
\cite{mai2020optimal}. The second approach 
is designed to deal with some of computational
issues of the first approach. Moreover, as we
will show, the optimal point of the first
approach can be computed from that of the
second approach.

\subsection{Convex Relaxation: M-Matrix Theory}
    \label{subsec:LowerBound}

Given $\boldsymbol{\lambda} \gneq \0$ and irreducible $B$, 
Theorem~\ref{thmEquilibrium} states that the unique 
equilibrium of \eqref{eq:pdot} which satisfies 
\eqref{eqSteadyState} is strictly positive. Hence,  
we can rewrite the constraints in \eqref{eqSteadyState}
as
\begin{align}
(\pb^{-1} - \1)\circ (\boldsymbol{\lambda} + B\pb) 
= \boldsymbol{\alpha}\circ \bs + \boldsymbol{\delta}, 
    \label{eqSteadyState2}
\end{align}
where $\pb^{-1} = (p_i^{-1} : \ i \in {\cal A})$.
By introducing a new variable 
\begin{align}
\zb := \pb^{-1}\circ(\boldsymbol{\lambda} + B\pb), 
    \label{eq_zp}
\end{align}
the constraint in \eqref{eqSteadyState2} can be rewritten as
\begin{equation}
\zb 
= \boldsymbol{\alpha}\circ \bs + \boldsymbol{\delta} 
	+ \boldsymbol{\lambda} + B\pb.
    \label{eqLinear_zsp}
\end{equation}
Note that \eqref{eqLinear_zsp}
is affine in $\zb$, $\bs$ and $\pb$, and the nonconvexity
in the equality constraint functions (mentioned
at the end of the previous section) is now captured by 
$\zb$, 
which from \eqref{eq_zp} can be expressed as
\begin{align}
(\diag(\zb) - B)\pb = \boldsymbol{\lambda}
\gneq \0. \label{eq_zp2}
\end{align}

\blue{We can show that the matrix $(\diag(\zb) - B)$ is a 
nonsingular M-matrix and, hence, $\pb = (\diag(\zb) - B)^{-1}\boldsymbol{\lambda}$ as follows: from (8), 
since $\boldsymbol{\lambda} \gneq \0$, we have
$\lambda_{i^*} > 0$ for some $i^*$. Since matrix $B$ is 
assumed irreducible, for any $j$ such that $\lambda_j = 
0$, we can find a finite sequence $(i_0 = j, i_1, i_2, 
\ldots, i_r = i^*)$ such that $[(\diag(\zb) - B) \pb]_{i^*} = 
\lambda_{i^*} > 0$ and $(\diag(\zb) - B)_{i_k, i_{k+1}} 
= - B_{i_k, i_{k+1}} \neq 0$ for all $k \in \{0, 1, 
\ldots, r-1\}$. Because $\pb > 0$, Lemma 1-(f) tells us 
that this is equivalent to matrix $(\diag(\zb) - B)$ 
being a nonsingular M-matrix. }
As a result, the original problem (P) can be reformulated as
\begin{align}
\hspace{-0.45in}
(\mathrm{P2})\qquad \quad  
\min_{\bs, \pb, \zb}&
	\quad f(\bs, \pb) \nnb \\
\mathsf{s.t.}&\quad \pb 
	= (\diag(\zb) - B)^{-1} \boldsymbol{\lambda} \nnb \\
&\quad \zb = \boldsymbol{\alpha}\circ \bs 
	+ \boldsymbol{\delta} + \boldsymbol{\lambda} 
		+ B\pb \nnb\\
&\quad \bs \in \mathbb{R}^N_+,\quad \pb \in \R^N_+,
	\quad \zb \in \Omega, \nnb
\end{align}
\noindent where 
\begin{equation}
\Omega := \big\{\zb \in \R^N_+~|~ \diag(\zb) 
	- B \in \mathbb{M}^{N \times N}_+\big\}.
    \label{eq:Omega}
\end{equation}

We can show that the set $\Omega$ in 
\eqref{eq:Omega} is convex. This is proved in 
Appendix~\ref{Appendix_CvxOmega}. 
Also, it follows from Lemma
\ref{lemConvexityInvMMatrix} 
that for any $1 \le i,j \le N$, the element 
$\big[\big( \diag(\zb) -B \big)^{-1}\big]_{i,j}$
is convex and (element-wise) decreasing in 
$\zb\in\Omega$. 
For these reasons, we obtain 
the following convex relaxation of (P2).
\begin{tcolorbox}[colback=white]
\vspace{-0.22in}
\begin{subequations}
\beqa
\hspace{-0.48in}
(\mathrm{P_{R1}})\qquad  
\min_{\bs, \pb, \zb}
	&&  f(\bs, \pb) \lb
\mathsf{s.t.}&&\pb \ge  
	(\diag(\zb) - B)^{-1}\boldsymbol{\lambda} 
	\label{eq_pz_relax}\\
&& \zb = \boldsymbol{\alpha} \circ \bs 
	+ \boldsymbol{\delta} 
	+ \boldsymbol{\lambda} + B\pb
	\label{eq_z_relax}\\
&& \bs \in \R_+^N,\quad \pb \le \1,
	\quad \zb \in \Omega.
	\nonumber
\eeqa
\end{subequations}
\vspace{-0.30in}
\end{tcolorbox}
\noindent 
This convex relaxation can be solved by numerical 
convex solvers to 
provide a lower bound on the optimal value of (P). 
Also, as shown in the following theorem, 
its optimal point also leads to a feasible 
point for problem (P).

\begin{thm}\label{thmConvexRelaxSolution}
Let $\bx^*_{\rm{R}} := (\bs^*_{\rm{R}}, \pb^*_{\rm{R}}, 
\zb^*_{\rm{R1}})$ denote an optimal point of 
$(\rm{P_{R1}})$ 
and $f^*$ the optimal value of $(\rm{P})$. 
Then,  we have
$$
f(\bs^*_{\rm{R}}, \pb^*_{\rm{R}}) 
\le f^* \le f(\tilde{\bs}(\bx^*_{\rm{R}}), 
\tilde{\pb}(\bx^*_{\rm{R}})), 
$$
where $(\tilde{\bs}(\bx^*_{\rm{R}}), 
\tilde{\pb}(\bx^*_{\rm{R}}))$ 
is a feasible point of problem $(\rm{P})$ given by
\begin{align*}
\tilde{\pb}(\bx^*_{\rm{R}}) 
&= (\diag(\zb^*_{\rm{R}}) - B)^{-1}\boldsymbol{\lambda}
\ \mbox{ and } \\
\tilde{\bs}(\bx^*_{\rm{R}}) 
&= \bs^*_{\rm{R}} 
+  \diag(\boldsymbol{\alpha}^{-1})B(\pb^*_{\rm{R}} - 
\tilde{\pb}(\bx^*_{\rm{R}})).
\end{align*}
\end{thm}
\begin{proof}
The first inequality is obvious because 
$(\rm{P_{R1}})$ is 
a convex relaxation of $({\rm P})$. 
For the second inequality, note that 
$(\tilde{\bs}(\bx^*_{\rm{R}}), \tilde{\pb}
(\bx^*_{\rm{R}}), \zb^*_{\rm{R}})$ is a feasible 
point for $(\rm{P_{R1}})$. Also, it satisfies 
\eqref{eq_pz_relax} with equality. Thus, it is a 
feasible point for problem $(\rm{P})$, proving 
the second inequality. 
\end{proof}

Clearly,  $\bx_{{\rm R}}^*$ solves $(\rm{P})$ 
if the inequality 
constraints in \eqref{eq_pz_relax} are all active at 
$\bx_{{\rm R}}^*$, 
which means $f(\bs^*_{\rm{R}}, \pb^*_{\rm{R}}) 
= f(\tilde{\bs}(\bx^*_{\rm{R}}), 
\tilde{\pb}(\bx^*_{\rm{R}}))$. Based on this, we 
can provide a following sufficient condition for  
convex relaxation $(\rm{P_{R1}})$ to be exact.

\begin{lemma}\label{corExactRelaxation}
The above 
convex relaxation $(\rm{P_{R1}})$ is exact if 
\begin{align} 
B\T \diag(\boldsymbol{\alpha}^{-1}) \nabla w(\bs) 
    \le  \cb  \ \mbox{ for all } \bs\ge \0. 
    \label{eqExactRelaxationCond2}
			\end{align}
\end{lemma}
\begin{proof}
Suppose $(\tilde{\bs}, \tilde{\pb})$ 
is the feasible point 
of $(\rm P)$ given in 
Theorem~\ref{thmConvexRelaxSolution}. 
Since $w$ is convex, we have 
$w(\tilde{\bs}) - w(\bs^*_{\rm{R}}) \le 
\nabla w(\tilde{\bs})\T (\tilde{\bs} - 
\bs^*_{\rm{R}}) = \nabla w(\tilde{\bs})\T 	
\diag(\boldsymbol{\alpha}^{-1})B(\pb^*_{\rm{R}} 
- \tilde{\pb})$. 
From this inequality, the gap $f(\tilde{\bs}, 
\tilde{\pb}) - f(\bs^*_{\rm{R}}, 
\pb^*_{\rm{R}}) \leq    
(\nabla w(\tilde{\bs})\T
\diag(\boldsymbol{\alpha}^{-1})B 
- \cb\T)(\pb^*_{\rm{R}} - \tilde{\pb})$. 
Under condition \eqref{eqExactRelaxationCond2}, 
together with $\pb^*_{\rm{R}} \ge 
\tilde{\pb}$, this 
gap is nonpositive. 
By Theorem~\ref{thmConvexRelaxSolution}, 
this gap must be zero. Thus, $(\rm{P_{R1}})$ is exact.
\end{proof}


\begin{rem}
(\emph{Sufficient condition for exact relaxation}) 
First, roughly speaking, the condition in 
\eqref{eqExactRelaxationCond2} means that 
when the infection costs ${\bf c}$ are 
sufficiently high, 
the convex relaxation $({\rm P_{R1}})$ is exact 
and we can find optimal security 
investments, i.e., a solution to $(\rm P)$, 
by solving $(\rm{P_{R1}})$ instead. 
\blue{The intuition behind this observation is the
following:
as ${\bf c}$ becomes larger, the second term in 
the objective function, namely ${\bf c}\T \pb$, becomes
more important and an optimal point tries to suppress
it by reducing $\pb$. However, since $\pb$
must satisfy the inequality in \eqref{eq_pz_relax}, 
it can only be reduced till the equality holds, 
which satisfies the constraint in problem (P2).}
Second, condition \eqref{eqExactRelaxationCond2} can 
be verified prior to solving the relaxed problem. This 
can be done easily if $w$ is a linear function or an 
upper bound on the gradient $\nabla w$ is known. 
Finally, even when the convex relaxation is not exact, 
$(\tilde{\bs}, \tilde{\pb})$ can still be used as a good initial
point for a local search algorithm, such as 
the RGM developed in Section~\ref{sec:UpperBound} below.
\end{rem}

\begin{rem}\label{remRelaxationIssues}
(\emph{Numerical issues of $(\rm{P_{R1}})$})  
Although $(\rm{P_{R1}})$ is a 
\blue{convex problem,} 
there are a few numerical challenges. 
First, the Jacobian of constraint 
functions in \eqref{eq_pz_relax}, which involves 
the derivative of inverse
matrix $(\diag(\zb) - B)^{-1}$, tends to 
be dense even when $B$ is sparse. 
Thus, off-the-shelf convex solvers may
not be suitable for large systems. 
		
Second, although the constraint set $\Omega$ 
for $\zb$ (defined in \eqref{eq:Omega}) is convex, 
\blue{it is not numerically easy to handle,} especially 
for large networks. This is because $\Omega$ 
is not closed and 
\blue{$(\rm{P_{R1}})$} 
becomes invalid outside $\Omega$. 
Thus, a numerical algorithm ought to stay inside
$\Omega$ and, for this reason, the nonsingularity 
of the M-matrix, diag$(\zb) - B$, should be 
ensured at every step. 
In general, it takes $O(N^3)$ to check if the matrix 
satisfies this condition \cite{pena2004stable}. 
The following approach can, however, alleviate 
the computational burden.  

\begin{itemize}
\item[s1] Starting at some $\zb_0 
\!\in\! \Omega$, solve $(\rm{P_{R1}})$ 
only with the constraint $\zb\in \R^N_+$. 
Then, check if the obtained solution 
$\bx^*_{{\rm R}}$ satisfies $\zb^*_{\rm{R}} 
\in \Omega$, 
If so, $\bx^*_{{\rm R}}$ solves $(\rm{P_{R1}}$). 
Otherwise, go to step s2. 

\item[s2] Choose a simpler subset 
$\tilde{\Omega} \subset \Omega$ and solve 
$(\rm{P_{R1}})$ subject to a stricter constraint
$\zb \in \tilde{\Omega}$. If $\zb^*_{\rm{R}}$
in $\bx^*_{{\rm R}}$ lies in 
int$(\tilde{\Omega})$, the solution is optimal for
$(\rm{P_{R1}})$; otherwise, construct a new 
$\tilde{\Omega}$ so that $\zb^*_{\rm{R}}$ 
belongs to the interior of new $\tilde{\Omega}$
and repeat.
Below, we propose 
an efficient way to choose the subset 
$\tilde{\Omega}$ that is more suitable for 
numerical algorithms.

\end{itemize}
\end{rem}

\subsubsection{Construction of Convex Subsets of 
$\Omega$}
\label{secFeasibleSubsets}

A key observation to constructing 
suitable subsets of $\Omega$ is that, in view of  
Lemmas~\ref{lem_M_Matrix} and \ref{lemConvexityInvMMatrix}, 
$\Omega$ can be expressed as 
$$	
\Omega 
= \textstyle\bigcup_{\underline{\zb} \in \partial \Omega} 
	\{\zb\in \R^N_+~|~\zb \gneq \underline{\zb} \}.
$$
Thus, for every $\underline{\zb} \in \partial 
\Omega$, 
$\tilde{\Omega}(\underline{\zb}) 
:= \{ \zb\in \R^N_+~|~\zb \gneq \underline{\zb} \}
\subset \Omega$.
Our goal is to find some $\check{\zb} \in \partial \Omega$ 
such that an optimal point $\bx^*_{\rm{R}}$ that solves 
the relaxed problem with $\Omega$ replaced by 
$\tilde{\Omega}(\check{\zb})$, satisfies $\zb^*_{\rm{R}} 
\in$ int $\tilde{\Omega}(\check{\zb})$. 
Below, we provide several possible choices for $\underline{\zb}$ 
with increasing computational complexity. 
		
\paragraph{Diagonal dominance} 
		
The matrix $\diag(\zb) - B$ is nonsingular if it is strictly diagonally dominant. This can be guaranteed by choosing 
$\underline{\zb} > B\1$, where the lower bound $B\1$ represents the total rate of infection from immediate neighbors in the graph $\mathcal{G}$. From \eqref{eq_z_relax}, a trivial sufficient
condition is $\boldsymbol{\delta} + \boldsymbol{\lambda}
\geq B\1$. But, we observe empirically that this often leads to suboptimal solutions. 
		
\paragraph{Dominant eigenvalue}
		
Another straightforward lower bound is given by $\underline{\zb}  > \rho(B)\1$. 
Recall that the spectral radius $\rho(B)$ is also an eigenvalue of $B$ and equal to $\bar{\sigma}(B)$, which can be computed efficiently using, for example, the power method.
		
\paragraph{Iterative dominant eigenvalue selection via matrix balancing}
\label{ApproxOmegaCombine}
		
Unfortunately, we observe empirically that a static selection 
of the subset $\tilde{\Omega}$ does not always lead to a good
solution and a following iterative algorithm yields better performance:
let $\hb >\0$ be a normal vector of the plane tangent to 
the closure of $\Omega$ at 
some $\underline{\zb} \in \partial \Omega$ such that
\begin{eqnarray}
\underline{\zb} 
\myeq \textstyle\arg\min_{\zb \in \R^N_+} \quad   
	\{\hb\T\zb ~|~ \zb \in {\rm cl} (\Omega) \}\nnb \\
\myeq \textstyle\arg\min_{\zb \in \R^N_+}\quad   
	\{\hb\T\zb ~|~ \underline{\sigma} 
	\big(\diag(\zb) - B\big) = 0 \}, 
	\label{MatBal}
\end{eqnarray}
where the second equality follows from the fact 
that we are 
minimizing a linear function over a closed convex set. 
The minimization in \eqref{MatBal} amounts to finding 
the smallest diagonal perturbation $\zb$ 
(in $1$-norm weighted 
by $\hb$) so that $B$ becomes (negative) stable. In 
Appendix~\ref{Appendix_MatrixBalancing}, 
we show that this is in fact a \emph{matrix balancing} 
problem, for which efficient algorithms exist 
(see \cite{ostrovsky2017matrix, boyd2004convex} 
for nearly-linear time centralized algorithms 
and \cite{mai2018distributed, mai2019asynchronous} for 
distributed algorithms with geometric convergence).

\begin{algorithm2e}[hbt]\label{algRelaxation}
\DontPrintSemicolon
\textbf{init}:  $t=0$, $\bar{h}>1$, 
$\underline{\zb}^{(0)}$ from \eqref{MatBal}\;
\While{stopping cond. not met}{
	$(\tilde{\bs}_{\rm{R}}^{(t+1)}\!, 
		\tilde{\pb}_{\rm{R}}^{(t+1)}\!, 
	\tilde{\zb}_{\rm{R}}^{(t+1)}) \gets$ solve 
		$(\rm{P_{R}})$~:~ 
	$\zb\in\!\tilde{\Omega}(\underline{\zb}^{(t)})$\;
	$\mathcal{I}_{ac} \gets \{i\in \mathcal{A}~|~ 
		[\tilde{\zb}_{\rm{R}}^{(t+1)}]_i 
			= [\underline{\zb}^{(t)}]_i\}$\; 
	\If{$\mathcal{I}_{ac} = \varnothing$}{
	    \textbf{break}\;\vspace{-1mm}
	    }
	$\hb^+ \gets (h_i^+ = 1, i\notin \mathcal{I}_{ac}; 
		h^+_i = \bar{h}, i \in \mathcal{I}_{ac})$\;
	$d \gets \underline{\sigma}\big( 
	\diag(\hb^+)^{-1}(\diag(\tilde{\zb}_{\rm{R}}^{(t+1)}) - B) 
	\big)$\;
	$\underline{\zb}^{(t+1)} \gets \tilde{\zb}_{\rm{R}}^{(t+1)} 
	- d \hb^+$\;
	$t\gets t+1$
	\vspace{-1mm}
}
\caption{\mbox{Algorithm for Convex Relaxation $(\rm{P_{R1}})$}}
\end{algorithm2e}

Our first proposed algorithm
(Algorithm~\ref{algRelaxation}) is based on 
the discussion in this subsection. 
Initially, we choose some $\bar{h} > 1$ and 
$\hb = \boldsymbol{\alpha}^{-1}\circ\nabla w(\sbold_0)$,
where $\sbold_0$ is the initial choice of security 
investments. This heuristic is based on the relaxed 
problem by weighting only the investment cost 
$w(\sbold)$ without considering $\pb$. 
Subsequent iterations 
are based on dominant eigenvalues with varying 
weights determined by $\hb^+$, which reflects active 
constraints of $\tilde{\zb}_{\rm{R}}$ (of the current
solution). Since $\tilde{\zb}_{\rm{R}} 
\in \Omega$, we have $\underline{\sigma}
\big(\diag(\tilde{\zb}_{\rm{R}}) - B\big) > 0$. Thus, 
we can construct a new subset 
$\tilde{\Omega}(\underline{\zb})$ by translating the set 
$\{ \zb \ge  \tilde{\zb}_{\rm{R}}\}$ towards the boundary 
$\partial\Omega$ in the direction of $\hb^+$, so that 
$\tilde{\zb}_{\rm{R}}$ lies in the interior of new
$\tilde{\Omega}(\underline{\zb})$. In our numerical 
studies (Section~\ref{sec:Numerical}), we use 
$\bar{h} = 10$. 

\blue{ 
Note that Algorithm~\ref{algRelaxation} is guaranteed to converge 
because the problem is convex and the objective function value
decreases after each iteration. 
Although we cannot provide a convergence rate, 
numerical studies in Section~\ref{sec:Numerical} show 
that only a few iterations are needed in most cases.%
}
	
\subsection{Convex Relaxation Based on Exponential Cones}
    \label{subsec:LowerBoundEXPcone}

As explained in the previous subsection, a possible
difficulty in solving the convex relaxation in 
$(\rm{P_{R1}})$ is taking into account two 
constraints -- constraint in
\eqref{eq_pz_relax} and $\zb \in \Omega$. 
Here, we present an alternative 
convex relaxation of the original problem, which avoids these issues by
introducing auxiliary optimization 
variables and relaxing the equality
constraint in \eqref{eqSteadyState}
without the need for the constraint
set $\Omega$. 

First, recall from the previous subsection that 
the constraint in \eqref{eqSteadyState} can be 
rewritten as 
\begin{align}
\pb^{-1}\circ\boldsymbol{\lambda} + \pb^{-1}\circ B\pb 
= \boldsymbol{\lambda} + B\pb 
+ \boldsymbol{\alpha}\circ \bs + \boldsymbol{\delta}. \tag{\ref{eqSteadyState2}}
\end{align} 
Since any solution must satisfy $\pb \in (0,1]^N$, we introduce following auxiliary variables and rewrite
the equality constraint in \eqref{eqSteadyState2}: 
for fixed $\yb \in \RN_+$, define
\begin{align} 
\hspace{-0.1in} \pb := e^{-\yb}, \ \tb := \blambda\circ e^{\yb}, \ 
U := \diag(e^{\yb})B\diag(e^{-\yb}).  \label{eqEXP_var}
\end{align}
Using these new variables, \eqref{eqSteadyState2} can 
be rewritten as follows.
\begin{align}
\tb + U\1 
= \boldsymbol{\lambda} + B\pb + \boldsymbol{\alpha}\circ \bs + \boldsymbol{\delta}
\label{eqSteadyState6}
\end{align}
Then, problem $({\rm P})$ is equivalent 
to the following problem. 
\begin{align}
\hspace{-0.6in}
(\mathrm{P3})\qquad \quad  
\min_{\bs\ge \0, \pb\ge \0, \yb, \tb, U}&
\quad f(\bs, \pb) = w(\bs) + \cb\T\pb \nnb \\
\mathsf{s.t.}&\quad \eqref{eqEXP_var}, \eqref{eqSteadyState6} \nnb 
\end{align}

The equivalent problem $({\rm P3})$ is still 
nonconvex due to the constraints in 
\eqref{eqEXP_var}. We can relax these 
equality constraints with the following 
inequality convex constraints.
\begin{align} 
\hspace{-0.1in}
\1 \ge \pb \ge e^{-\yb}, \  \tb \ge  \blambda\!\circ\! e^{\yb}, 
\ U \ge \diag(e^{\yb})B\diag(e^{-\yb})
    \label{eqEXP_cone}
\end{align}
This leads to the following second convex relaxation.
\begin{tcolorbox}[colback=white]
\vspace{-0.21in}
\begin{align}
\hspace{-0.5in}
(\mathrm{P_{R2}})\qquad   
\min_{\bs \ge \0, \pb, \yb \ge \0, \tb, U}&
\quad f(\bs, \pb) = w(\bs) + \cb\T\pb \nnb \\
\mathsf{s.t.}&\quad \eqref{eqSteadyState6}, \eqref{eqEXP_cone}  \nnb 
\end{align}
\vspace{-0.28in}
\end{tcolorbox}

We can express the constraints in \eqref{eqEXP_cone} as a following set of at most 
$2N+m$ exponential cone constraints: 
\begin{subequations}
\begin{eqnarray} 
\hspace{-0.45in} 
    (p_i, 1, -y_i)\in \mathcal{K}_{\rm exp}
    && \myb 
    \mbox{ for all } i \in \mathcal{A}
    \label{eqEXP_P_cone}\\
\hspace{-0.45in} 
    (t_i, 1, y_i+\log \lambda_i)\in \mathcal{K}_{\rm exp}
    && \myb 
    \mbox{ for all } i \in \Psi_{\boldsymbol{\lambda}} 
    \label{eqEXP_T_cone}\\
\hspace{-0.45in} 
    (u_{ij}, 1, y_i - y_j +\log b_{ij})\in 
    \mathcal{K}_{\rm exp} 
    && \myb 
    \mbox{ for all } (i,j) \in 
    \mathcal{E} 
    \label{eqEXP_U_cone}
\end{eqnarray}
\end{subequations}
where $\mathcal{K}_{\rm exp} 
\!:=\! {\rm cl}(\{ (x_1, x_2, x_3)
\,|\, x_1 \!\ge\! x_2 e^{x_3/x_2}, x_2 \!>\! 0 \})$, 
and $\Psi_{\boldsymbol{\lambda}} := 
\{ i \in \mathcal{A} ~|~\lambda_i >0 \}.$ 
These constraints can be handled efficiently 
by 
conic optimization solvers, e.g., 
MOSEK \cite{mosek}.

\begin{rem}
We demonstrate below that, somewhat surprisingly, the
convex relaxations in $({\rm P_{R1}})$ and 
$({\rm P_{R2}})$ are in fact equivalent. Moreover,
although one may suspect that the size of 
$(\rm{P_{R2}})$ with $4N+m$ variables and $3N+m$ 
constraints is much larger than the size of 
$(\rm{P_{R1}})$, the constraints of $({\rm P_{R2}})$ 
are much easier to handle numerically. We will 
provide numerical results to illustrate this in Section~\ref{sec:Numerical}.
\end{rem}

Analogously to Theorem~\ref{thmConvexRelaxSolution}, 
the following theorem tells us how to find a 
feasible point of the problem ${\rm (P)}$, using 
an optimal point of problem $({\rm P_{R2}})$. In 
addition, it asserts that the two convex relaxations
$({\rm P_{R1}})$ and $({\rm P_{R2}})$ are equivalent 
in that their optimal values coincide and we can 
find an optimal point of $({\rm P_{R1}})$ from 
an optimal point of $({\rm P_{R2}})$
.

\begin{thm}\label{thmConvexRelaxSolution_Cone}
Suppose $\bx^+_{\rm R} := (\bs^+, \pb^+, \yb^+, 
\tb^+, U^+)$ is an optimal point of $(\rm{P_{R2}})$.
Then, we have
\beqa
f(\bs^+, \pb^+) 
\le f^* \le f(\bs', \pb'), 
    \label{eq:thm3}
\eeqa
where $f^*$ and $(\bs', \pb')$ are the optimal 
value and a feasible point, respectively, of the 
original problem (P) with 
\[
\pb' = e^{-\yb^+} \ \mbox{ and } \
\bs' = \bs^+ + \diag(\boldsymbol{\alpha}^{-1})
        B(\pb^+ - \pb').
\]            
Moreover, the last two constraints of 
\eqref{eqEXP_cone} are active at 
$\bx^+$, i.e.,  
\begin{align} 
\tb^+ = \blambda\circ e^{\yb^+}
    \mbox{ and } 
     U^+ = \diag(e^{\yb^+})B\diag(e^{-\yb^+}) . 
     \label{eqEXP_varopt}
\end{align}
Finally, $(\bs^+, \pb^+,\tb^+ \!+ U^+\1)$
is an optimal point of 
$(\rm{P_{R1}})$.
\end{thm}
\begin{proof} 
Please see Appendix \ref{appen_exact_relax_cone}.  
\end{proof}

As a direct consequence of the theorem, 
the sub-optimality of $(\bs', \pb')$ can be assessed 
using the gap $f(\bs', \pb') - f(\bs^+, \pb^+)$. 
Similar to Lemma~\ref{corExactRelaxation}, 
the condition in \eqref{eqExactRelaxationCond2} 
provides a sufficient condition for this gap to 
be zero, i.e.,  
the convex relaxation in $({\rm P_{R2}})$ is
exact.

\section{Upper bounds on Optimal Value}
    \label{sec:UpperBound}

The previous section described (i) how we can 
formulate a convex relaxation of problem $(\rm P)$, 
which provides a lower bound on the optimal value
of $(\rm P)$, using two different techniques and (ii)
how to find a feasible solution to $(\rm P)$ using an 
optimal point of a convex relaxation. 

Although the convex relaxation $(\rm{P_{R1}})$ 
or $({\rm P_{R2}})$ may 
be exact under certain conditions, this is not 
true in general.
In addition, $({\rm P_{R1}})$ may not scale 
well due to the 
constraint in \eqref{eq_pz_relax}; see also 
Remark~\ref{remRelaxationIssues} above and numerical 
results in Section~\ref{sec:Numerical}. For these 
reasons, we also propose efficient algorithms 
for finding a local minimizer of the
nonconvex problem $({\rm P})$ in this section. 
These algorithms provide an 
upper bound on the optimal value, which, together 
with the optimal value of a convex relaxation 
when available, can be used to offer a bound on 
the optimality gap.

\subsection{Reduced Gradient Method}
    \label{sec:RGM}

Among different nonconvex optimization approaches, 
we first choose the 
RGM \cite{lasdon1974nonlinear,
gabay1976efficiently} because it is well suited to the problem $({\rm P})$ 
and, more importantly, is scalable. 

\subsubsection{Main Algorithm}
	
First, together with Theorem~\ref{thmEquilibrium}, 
the implicit function theorem tells us that the 
condition $\gb (\bs, \pb) = \0$ in 
\eqref{eqSteadyState} defines a continuous mapping 
$\pb^*: \sbold \in\R^N_+ \mapsto \pb^*(\sbold) 
\in (0,1)^N$ such that $\gb(\bs, \pb^*(\bs)) = 0$. 
Thus, problem $({\rm P})$ can be transformed 
to a reduced problem 
only with optimization variables $\sbold$:
\begin{align}
\min_{\bs \in \mathbb{R}^N_+}&\quad 
    F(\bs):= w(\bs) + \cb\T\pb^*(\bs).
    \label{eqObjectiveFun_S}
\end{align}
Suppose that $(\bs^\star, \pb^\star)$ is a feasible point 
of $({\rm P})$. Then, the gradient of $F$ at $\bs^\star$ 
is equal to
\begin{equation*} 
\nabla F(\bs^\star) 
= \nabla w(\bs^\star) + J(\bs^\star)\T\cb,
\end{equation*}
where $J(\bs^\star) = \big[ \partial p^*_i(\bs^\star) 
/ \partial s_j \big]$. This matrix can be computed 
by  totally differentiating $\gb (\bs, \pb^*(\bs)) 
= \0$ at $\bs^\star$: 
the calculation of total derivative yields 
\begin{equation}    \label{eq:MJ=p}
M(\bs^\star)J(\bs^\star) 
= -\diag(\boldsymbol{\alpha} \circ\pb^\star)
\end{equation}
with $M(\bs^\star) 
= \diag(\boldsymbol{\alpha}\circ\bs^\star 
	+ \boldsymbol{\delta} 
    + \boldsymbol{\lambda} + B \pb^\star) 
    - \diag(\1 - \pb^\star) B$. 
The following lemma shows that $M(\bs^\star)$ is 
nonsingular.
\begin{lemma}   \label{lemma_M_M_matrix}
	The matrix $M(\bs^\star)$ is a nonsingular M-matrix.
\end{lemma}
\begin{proof}
	First, note that $M(\bs^\star)$ is a Z-matrix. 
	Second, after some algebra, the constraint 
	$\gb(\bs^\star, \pb^\star) = \0$ is equivalent to
	$M(\bs^\star)\pb^\star 
	= \boldsymbol{\lambda} 
	    + \pb^\star \circ (B \pb^\star)$. 
	Since $\pb^\star > \0$, we have $\boldsymbol{\lambda} 
	+ \pb^\star \circ (B \pb^\star) > \0$. Thus, 
	Lemma~\ref{lem_M_Matrix}-(g) implies that $M(\bs^\star)$ 
	is a nonsingular M-matrix. 
\end{proof}

\blue{As a result}, 
$J(\bs^\star) = - M(\bs^\star)^{-1} \diag(\boldsymbol{\alpha}\circ\pb^\star)$
from \eqref{eq:MJ=p} and the gradient of $F$ is given by 
\begin{equation*} 
\nabla F(\bs^\star)
= \nabla w(\bs^\star) 
    - \boldsymbol{\alpha}\circ\pb^\star \circ \big( \big( M(\bs^\star)\big)\nT \cb \big).
\end{equation*}
Hence, we can apply the \blue{(projected)} gradient descent 
method on the reduced problem in  \blue{\eqref{eqObjectiveFun_S}. 
For instance, \cite[Proposition~2.3.3]{Bertsekas99Book} shows that this method converges to a stationary point under step sizes 
$\{\gamma_t\}_{t\ge 0}$ chosen by the Armijo backtracking line search.}

Note that, after each update of $\bs$ during a search, 
we need to compute the corresponding $\pb$ so that 
$(\bs, \bp)$ is feasible for the problem $({\rm P})$. 
As mentioned earlier,
this can be done by using the fixed point iteration 
in \eqref{eqFixedPointEquilibrium}. 

Our proposed algorithm is provided in Algorithm
\ref{algReduceGrad}.

\begin{algorithm2e}[hbt]\label{algReduceGrad}
	\DontPrintSemicolon
	\textbf{init}:  $t=0$, feasible $(\bs^{(0)}, \pb^{(0)})$\;
	\While{stopping cond. not met}{
		\mbox{\!$M^{(t)} \!\!\gets \! \diag(\boldsymbol{\alpha} \!\circ\!\bs^{(t)} \!+\!\boldsymbol{\delta}\!+\! \boldsymbol\lambda \!+\!  B\pb^{(t)} ) \!-\! \diag(\1\!-\!\pb^{(t)})B$}\;
		\vspace{-4mm}
		$\uu \gets  (M^{(t)})\nT\cb$\;
		$\gamma_t \gets$ \textsc{line\_search}\;
		$\bs^{(t+1)} \gets  \big[ \bs^{(t)} - \gamma_t \big( \nabla w(\bs^{(t)}) - \boldsymbol{\alpha}\circ\pb^{(t)}\circ \uu\big) \big]_+$\;
		$\pb^{(t+1)} \gets \pb^*(\bs^{(t+1)})$ using \eqref{eqFixedPointEquilibrium}\;
		$t\gets t+1$\;\vspace{-1mm}
	}
	\caption{Reduced Gradient Method}%
\end{algorithm2e}
\vspace{-0.1in}

\subsubsection{Computational Complexity and Issues}
	\label{subsec:ComputationalComp}

For large systems, a naive evaluation of the gradient 
$\nabla F$, which requires the inverse matrix 
$\big( M(\bs^\star) \big)\nT$, becomes computationally 
expensive, if not infeasible. For this reason, we develop 
an efficient subroutine for computing $\nabla F$. This is
possible because our algorithm only requires 
$\uu$ (in line~4 of Algorithm~\ref{algReduceGrad}), not 
$\big( M(\bs^\star) \big)\nT$. 

For fixed $t \in \N := \{0, 1, \ldots\}$, 
the vector $\uu$ is a solution to 
a set of linear equations $M\T \uu= \cc$, where the 
matrix $M\T$ tends to be sparse for most real graphs
$\mathcal{G}$. Thus, there are several efficient algorithms
for solving them. In this paper, we employ the power 
method: let $M = D - E$, where $D$ and $E$ denote the 
diagonal part and off-diagonal part of $M$, 
respectively. Then, the linear equations are 
equivalent to $\cc =  D\uu - E\T\uu$.
Since $D$ is invertible, the following fixed point 
relation holds: 
\begin{equation}
\uu = D^{-1}E\T\uu + D^{-1}\cc 
    =: G(\uu). \label{eqFixedPointGradient}
\end{equation}
As $M \in \mathbb{M}^{N \times N}_+$ (Lemma
\ref{lemma_M_M_matrix}), Lemma~\ref{lem_M_Matrix}-(d) 
tells us that $M = D - E$ is a convergent 
splitting and the mapping $G$ in 
\eqref{eqFixedPointGradient} is a contraction mapping 
with coefficient $\rho(D^{-1}E\T) < 1$. Hence, the 
iteration $\uu_{k+1} = G(\uu_k)$ converges to the 
solution $\uu$ exponentially fast. 
Moreover, this iteration 
is highly scalable 
\blue{
because $E = \diag(\1-\pb)B$ 
is sparse, requiring only $O(|\mathcal{E}|)$ 
memory space and $O(|\mathcal{E}|)$ operations per iteration.}

\subsection{Sequential Convex Programming Method}
    \label{subsec:SCM}

In this subsection, we will develop a second efficient 
algorithm for finding a {\rm local} minimizer of 
the original nonconvex problem $({\rm P})$. This 
will provide another upper bound 
on the optimal value we can use to provide an 
optimality gap together with a lower bound from 
convex relaxations.

Our algorithm is based on SCP
applied to the original formulation in \eqref{eqObjectiveFun}--\eqref{eqSteadyState}. To this 
end, we successively convexify the constraint 
\eqref{eqSteadyState} using 
first order approximations. 
The novelty of our approach 
is to linearize (only) the terms $p_ip_j$ and then 
employ either the convexity result in 
Lemma~\ref{lemConvexityInvMMatrix} 
or the exponential cone formulation as in Section~\ref{sec:LowerBound}.


At each iteration $t \in \N$, we 
replace the terms $p_ip_j$ with their first order Taylor 
expansion at $\pb^{(t)}$, resulting in the following 
partially linearized equality constraint functions: 
\begin{eqnarray}
\gb^{(t)}(\bs, \pb) 
\myeq \boldsymbol\lambda +  \pb^{(t)}\circ (B\pb^{(t)}) 
    - \big( \boldsymbol{\lambda} 
    + B\pb^{(t)} \big)\circ \pb \label{eqLinearize_i} \\
&& + (\1-\pb^{(t)})\circ(B\pb) - (\boldsymbol{\alpha}\circ\bs 
    + \boldsymbol{\delta})\circ \pb.
    \nonumber
\end{eqnarray}
The partial linearization error 
\blue{is equal to 
$\gb^{(t)}(\bs, \pb) - \gb(\bs, \pb) = (\pb-\pb^{(t)})\circ B(\pb-\pb^{(t)}).$}   
When $\pb$ is close to $\pb^{(t)}$, 
this error will likely be `small' and we 
expect the linearization step to be acceptable. 
This allows us to approximate the constraint 
\eqref{eqSteadyState} with $\gb^{(t)}(\bs, \pb) = \0$,
which can be rewritten as
\begin{align}
\big( L^{(t)} + \diag(\boldsymbol{\alpha}\circ\bs) \big)\pb = \boldsymbol{\lambda}^{(t)}, \label{eqSS_linearized}
\end{align}
where $\boldsymbol{\lambda}^{(t)} = \boldsymbol\lambda +  \pb^{(t)}\circ(B\pb^{(t)})$, and 
\begin{align*}
L^{(t)} &=  \diag(\boldsymbol{\delta}+ \boldsymbol\lambda +  B\pb^{(t)} ) - \diag(\1-\pb^{(t)})B .
\end{align*}

This gives us the following subproblem we need to 
solve at each iteration $t \in \N$.
\begin{eqnarray}\label{eq:subprob1}
\hspace{-0.4in} {\rm (S1)} \hspace{0.4in}
\min_{\bs\ge \0, \pb\ge \0} \quad \big\{  w(\bs) + \cb\T\pb 
~|~\eqref{eqSS_linearized}~\text{holds} \big\}
\end{eqnarray}
The solution at the $t$-th iteration is then
used to construct a new constraint for the $(t+1)$-th
iteration, and we repeat this procedure until some
stopping condition is met. 
Unfortunately, 
the problem in \eqref{eq:subprob1} is still nonconvex.
But, as we show below, under certain conditions, 
it can be transformed 
to a convex problem, which can be solved efficiently.

\subsubsection{Convex Formulation Based on M-matrix}
    \label{subsubsec:M-matrix}

First, we show that if $(\bs, \pb)$ is feasible for 
the subproblem in \eqref{eq:subprob1}, then $L^{(t)} + 
\diag(\boldsymbol{\alpha}\circ\bs) \in \mathbb{M}_+^{N\times N}$: note that this is always a Z-matrix. 
Also, from \eqref{eqLinearize_i} and \eqref{eqSS_linearized}, any feasible $\pb$ 
for the subproblem must be positive,
and given $\pb^{(t)} > \0$, we have $\boldsymbol{\lambda}^{(t)}
> \0$ from its definition. Because $\pb$ and 
$\boldsymbol{\lambda}^{(t)}$ are positive, together with 
condition \eqref{eqSS_linearized}, Lemma~\ref{lem_M_Matrix}-(g)
implies that $L^{(t)} \!+\! \diag(\boldsymbol{\alpha}\circ\bs) \!\in\! \mathbb{M}_+^{N\times N}$. This in turn tells us from 
\mbox{Lemma~\ref{lem_M_Matrix}-(b)} that its inverse exists and is 
nonnegative. 
Consequently,
\begin{eqnarray}
\pb 
= \big( L^{(t)} + \diag(\boldsymbol{\alpha}\circ\bs) \big)^{-1} 
    \boldsymbol{\lambda}^{(t)} 
> \0 .
    \label{eq:subprob2}
\end{eqnarray}

This allows us to reformulate the subproblem in 
\eqref{eq:subprob1} as follows: we 
replace $\pb$ with the above expression in 
\eqref{eq:subprob2} and introduce a new constraint that $\bs$ 
belongs to a feasible set
\[
\Omega^{(t)} = \big\{\bs \in
\RN_+ \ | \ L^{(t)} + \diag(\boldsymbol{\alpha}\circ\bs)\in \mathbb{M}^{N\times N}_+ \big\}.
\]
Note that $\Omega^{(t)}$ is convex (the proof follows 
similar arguments in Appendix~\ref{Appendix_CvxOmega}). 
The partially linearized 
subproblem in \eqref{eq:subprob1} can now be 
written as
\begin{equation}\label{SCP_M_Mat}
\hspace{-0.8in} ({\rm P_L}) \hspace{0.5in}  
\displaystyle \min_{\bs \in \Omega^{(t)}}\quad J(\bs) 
    := w(\bs) + \zeta(\bs),
\end{equation}
where 
$\zeta(\bs) 
:= \cb\T\big( L^{(t)} 
    + \diag(\boldsymbol{\alpha}\circ\bs)\big)^{-1}
        \boldsymbol{\lambda}^{(t)}$
is a convex and decreasing function on $\Omega^{(t)}$ 
in view of  Lemma~\ref{lemConvexityInvMMatrix}. 
Thus, the problem $({\rm P_L})$ is convex at every 
iteration $t$. Note that, when solving $({\rm P_L})$, 
we need to ensure the constraint $\bs \in \Omega^{(t)}$ is
satisfied. This can be done in a manner similar to 
that discussed in subsection~\ref{secFeasibleSubsets}.

After computing an optimal point of $({\rm P_L})$ at the 
$t$-th iteration, which we denote by $\bs^{(t+1)}$, we then 
find $\pb^{(t+1)}$ satisfying $\gb(\bs^{(t+1)}, \pb^{(t+1)})
= \0$ (constraint \eqref{eqSteadyState}) using the fixed 
point iteration in \eqref{eqFixedPointEquilibrium}. 
Thus, we obtain a feasible solution 
$(\bs^{(t+1)}, \pb^{(t+1)})$ to the original problem 
$({\rm P})$ after each iteration $t \in \N$. 
The proposed algorithm based on this approach is provided 
in Algorithm~\ref{algSCP} below. 
\begin{algorithm2e}[hbt]\label{algSCP}
	\DontPrintSemicolon
	\textbf{init}:  $t=0$, $\pb^{(0)} \in [0,1]^N$\;
	\While{stopping cond. not met}{
		$\boldsymbol{\lambda}^{(t)} \gets  \boldsymbol\lambda +  \pb^{(t)}\circ(B\pb^{(t)})$\;
		$L^{(t)} \gets   \diag(\boldsymbol{\delta}+ \boldsymbol\lambda +  B\pb^{(t)} ) - \diag(\1-\pb^{(t)})B$\;
		\mbox{$\bs^{(t+1)} \!\gets\! \displaystyle\arg\!\!\min_{\bs \in \Omega^{(t)}} w(\bs) \!+\! \cb\T  \big( L^{(t)} \!\!+\! \diag(\boldsymbol{\alpha}\circ\bs) \big)^{-1}\!\blambda^{(t)}$}\;
		\vspace{-4mm}
		$\pb^{(t+1)} \gets  \pb^*(\bs^{(t+1)})$ using \eqref{eqFixedPointEquilibrium}\;
		$t\gets t+1$\;\vspace{-1mm}
	}
	\caption{Sequential Convex Programming}%
\end{algorithm2e}
\vspace{-0.1in}

\begin{rem}(\emph{Complexity}) 
For small and medium-sized networks, the subproblem can be 
solved using off-the-shelf numerical convex solvers, e.g., 
interior point methods. In this paper, we use an 
interior-point method to solve the subproblem 
$({\rm P_L})$, 
which employs the Newton's algorithm on a sequence of 
equality constrained problems. Since the number of 
variables is $O(N)$ and the number of constraints is 
also $O(N)$, the worst case 
complexity is $O(N^3)$ \cite{boyd2004convex}. 

For large networks, we take advantage of the fact
that we do not need to solve $({\rm P_L})$ exactly
at each iteration. Thus, we can use simple approximations
of $\Omega^{(t)}$ and employ computationally cheaper methods
to solve $({\rm P_L})$. For example, we can follow the 
same gradient-based approach in \cite{mai2018optimizing} 
for solving the subproblem, where the gradient of 
$\zeta(\bs)$ given by
\[
\nabla \zeta(\bs) \!=\! -(S^{-\mathsf{T}}\cb)\!\circ\!\boldsymbol{\alpha}\!\circ\!(S^{-1}\boldsymbol{\lambda}^{(t)}),
\]
where $S = L^{(t)} + \diag(\boldsymbol{\alpha}\circ\bs)$,
can be computed efficiently using the power method as 
explained in subsection~\ref{subsec:ComputationalComp}. 
\end{rem}

\subsubsection{{Convex Formulation Based on Exponential Cones}}

As discussed above, if $(\bs, \pb)$ is feasible 
for problem (S1), $L^{(t)} + 
\diag(\boldsymbol{\alpha}\circ\bs)$ is a nonsingular 
M-matrix and $\pb > \0$; see also~\eqref{eq:subprob2}. 
Thus, we can introduce a new variable $\yb$ satisfying 
\begin{equation*} 
\hspace{-0.1in} \pb = e^{-\yb}.
\end{equation*}
Then, \eqref{eqSS_linearized} becomes
$\big( L^{(t)} \!+\! \diag(\boldsymbol{\alpha}\circ\bs) \big)e^{-\yb} \!=\! \boldsymbol{\lambda}^{(t)}$,
which, after left-multiplying both sides by $\diag(e^{\yb})$, \mbox{is equivalent to}
\begin{align}
\boldsymbol{\alpha}\circ\bs + \boldsymbol{\delta}+ \boldsymbol\lambda +  B\pb^{(t)}  =  \diag(e^{\yb})B^{(t)}e^{-\yb} + \diag(e^{\yb})\boldsymbol{\lambda}^{(t)},\nnb
\end{align}
where $B^{(t)} = \diag(\1-\pb^{(t)})B$. As a result, the subproblem is equivalent to the following problem.
\begin{align}
\hspace{-0.4in} {\rm (S2)} \hspace{0.25in}
\min_{\bs\ge \0, \yb, \tb, U}&
\quad f^{(t)} = w(\bs) + \cb\T e^{-\yb} \nnb \\
\mathsf{s.t.}
&\quad \tb + U\1 = \boldsymbol{\alpha}\circ\bs + \boldsymbol{\delta}+ \boldsymbol\lambda +  B\pb^{(t)} \nnb \\
&\quad \tb = \blambda^{(t)}\!\circ\! e^{\yb}\nnb\\ 
&\quad U = \diag(e^{\yb})B^{(t)}\diag(e^{-\yb}) \nnb
\end{align}
A convex relaxation of ${\rm (S2)}$ can be obtained
by replacing the last two equality constraints
with inequality constraints.
\begin{subequations}  \label{SCP_ExpC}
\begin{align}
\hspace{-0.1in} ({\rm S_{R1}}) \hspace{0.2in}
\min_{\bs\ge \0, \yb, \tb, U}&
\quad f^{(t)} = w(\bs) + \cb\T e^{-\yb} 
    \nonumber \\
\mathsf{s.t.}
&\quad \tb + U\1 
= \boldsymbol{\alpha}\circ\bs + \boldsymbol{\delta}
    + \boldsymbol\lambda +  B\pb^{(t)} 
    \label{SCP_ExpC_Linear} \\
&\quad \tb \ge \blambda^{(t)} \circ e^{\yb}
    \label{SCP_ExpC_t}\\
&\quad U \ge \diag(e^{\yb}) B^{(t)} \diag(e^{-\yb})
    \label{SCP_ExpC_U}
\end{align}
\end{subequations}
It turns out that this convex relaxation
is always exact. 

\begin{thm}\label{thmSCP_Cone_Exact}
The convex relaxation $({\rm S_{R1}})$ is exact.
\end{thm}
\begin{proof}
A proof can be found in Appendix
\ref{appen_exact_SCP_cone}. 
\end{proof}

We end this subsection by noting that this convex formulation is in fact equivalent to the one based on M-matrix in \eqref{SCP_M_Mat}. This can be shown using similar arguments used in the proof of Theorem~\ref{thmConvexRelaxSolution_Cone} and is omitted here.

\section{Special Case: $\boldsymbol{\lambda} = \0$} \label{sec:Case_lambda0}

In practice, we expect that the systems experience 
primary attacks infrequently and $\boldsymbol{\lambda}$
is small, and that steady-state infection probabilities
are not large.  
For this reason, we consider a limit case of our problem 
as $\boldsymbol{\lambda} \to \0$ with diminishing primary
attack rates. As we show, studying the special case
with $\boldsymbol{\lambda} = \0$ reveals additional insights 
into the steady-state behavior and provides an approximate 
upper bound on the system cost when $\boldsymbol{\lambda} 
\approx \0$, which can be computed easily.

This case reduces to a problem that has been studied 
by previous works, in which the adjustable curing rate
is equal to $\delta_i / q_i(s_i)$ for each 
$i \in \cA$.\footnote{In the previous studies
\cite{Mieghem-TON09, mai2018distributed,
Ottaviano-JCN18}, security investments affect the curing rates
rather than the breach probability, i.e., they determine
how quickly each system can recover from an infection, but
do not change the infection probability of systems.} 
A key difference between this case and when $\boldsymbol{\lambda} 
\gneq {\bf 0}$ is that Theorem~\ref{thmEquilibrium} cannot 
be applied to guarantee the uniqueness of an equilibrium 
because the assumption $\boldsymbol{\lambda}
\gneq {\bf 0}$ is violated. It turns out that this 
difference has significant effects on our problem, 
as it will be clear.

\subsection{Preliminary}
	\label{subsec:Prelim}
	
In the absence of primary attacks, if 
no system is infected at the beginning, obviously
they will remain at the state. However, 
if some systems are infected initially, there are two
possible outcomes based on the security investments
$\bs$. 
\\ \vspace{-0.12in}

{\bf Case 1: $\rho\big( {\rm diag}(\boldsymbol{\alpha}
\circ \bs + \boldsymbol{\delta})^{-1} B \big) \leq 1$ --} 
In this case, the unique (stable) equilibrium 
of \eqref{eq:pdot} is $\pb_{\rm se}(\bs) 
= \0$. Thus, 
as $t \to \infty$, $\pb(t) \to \0$ and 
all systems become free of infection. 

{\bf Case 2: $\rho\big( {\rm diag}(\boldsymbol{\alpha}
\circ \bs + \boldsymbol{\delta})^{-1} B \big) > 1$ --} 
In this case,
there are two equilibria of 
\eqref{eq:pdot} -- one stable equilibrium 
$\pb_{\rm se}(\bs) > \0$ and
one unstable equilibrium $\0$: (a) if $\pb(0) \neq 
\0$, although there are no primary attacks, 
we have $\pb(t) \to \pb_{\rm se}(\bs)$. 
As a result, somewhat surprisingly, 
infections continue to transmit 
among the systems indefinitely and do not go away;
and (b) if $\pb(0) = \0$, obviously $\pb(t) = \0$
for all $t \in \R_+$.
\\ \vspace{-0.1in}

Based on this observation, 
we define a function $\pb_{\rm se}:\R^N_+ \to 
[0, 1]^N$, where $\pb_{\rm se}(\bs)$ is the 
aforementioned stable equilibrium
of \eqref{eq:pdot} for the given security investment
vector $\bs \in \R^N_+$. 
It is shown in Appendix~\ref{appen:Continuity} 
that  $\pb_{\rm se}$ is a 
continuous function over $\R^N_+$.
This tells us that, if we start with $\pb(0) \neq 
{\bf 0}$, for any given security investments $\bs 
\geq {\bf 0}$, our steady-state cost is given 
by $w(\bs) + \cb\T \pb_{\rm se}(\bs)$. 
For this reason, we are interested
in  the following optimization problem.
\begin{tcolorbox}[colback=white]
\vspace{-0.18in}
\begin{align}
\min_{\bs \geq {\bf 0}} \Big\{w(\bs) 
    + \cb\T\pb_{\rm se}(\bs)\Big\}
    \label{eqCase0_prob}
\end{align}
\vspace{-0.23in}
\end{tcolorbox}
\noindent 
We denote the optimal value and the optimal set of 
\eqref{eqCase0_prob} by $f^*_0$ and $\mathbb{S}^*_0$, 
respectively. Based on the discussion, we have the 
following simple observation. 

\begin{thm}\label{thmSpectralRadiusOptimal}
If $\rho\big( \diag(\boldsymbol{\delta})^{-1} B \big) \!\le\! 1$, then $\bs^*\!=\!\0$ is the optimal point. Otherwise, $\rho\big( \diag(\boldsymbol{\alpha}\circ \bs^*+\boldsymbol{\delta})^{-1} B \big) \!\ge\! 1$ for all $\bs^*\in \mathbb{S}^*_0$. 
\end{thm}
\begin{proof}
The theorem follows directly from the above discussion and 
the monotonicity of the spectral radius of nonnegative 
matrices \cite[Thm~8.1.18]{horn85matrix}: if $A, B \in 
\mathbb{R}^{n\times n}_+$ such that $A\ge B$, then 
$\rho(A) \ge \rho(B)$.
\end{proof}
\begin{rem}
The theorem rules out the case where $\rho\big( \diag(\boldsymbol{\alpha}\circ \bs^*+\boldsymbol{\delta})^{-1} B \big) < 1$ for some $\bs^*\!\neq\!\0$. 
It suggests that if the recovery rates of all the systems 
are sufficiently large, no additional investments are 
needed. Otherwise, at any solution $\bs^* \in 
\mathbb{S}^*_0$, the spectral radius is either 
(i) at the threshold (of one) or (ii) strictly above 
the threshold. In case (i), $w(\bs^*)$ is also 
the smallest investment cost to suppress the 
spread in that $\lim_{t \to\infty} \pb(t) = \0$ 
for all $\pb(0)$. We will show in 
subsection~\ref{subsec:Case0_main} below how to compute this 
minimum investment for suppression, denoted by $C^*$.  
On the other hand, case (ii) corresponds to an endemic state, i.e., $\lim_{t\to\infty} \pb(t) = \pb_{\rm se}(\bs^*) >  \0$ when $\pb(0) \neq \0$. 
In this case, we will demonstrate that we can find upper and lower bounds on the optimal cost using our techniques in \mbox{Sections~\ref{sec:LowerBound} and \ref{sec:UpperBound}}.
\end{rem}

In order to facilitate our discussion, we introduce
following related optimization problems with a
{\em fictitious} constraint on security investments.
\blue{The goal of imposing a fictitious budget constraint
is not to investigate a problem with a budget constraint; 
instead, it is used to facilitate the determination of 
a (nearly) optimal point of \eqref{eqCase0_prob} as we will
show.}
\begin{tcolorbox}[colback=white]
\vspace{-0.18in}
\begin{eqnarray} \label{eq:constr1}
\min_{\bs \ge \0}\quad \Big\{ w(\bs) 
    + \cb\T \pb_{\rm se}(\bs) ~|~ w(\bs) \leq C \Big\} 
\end{eqnarray}
\vspace{-0.23in}
\end{tcolorbox}

We define a function $f_0: \R_+ \to \R_+$, where
$f_0(C)$ is the optimal value of the above optimization 
problem for
a given budget $C$. Clearly, the 
function $f_0$ is continuous and nonincreasing, and 
$\lim_{C\to \infty} f_0(C) = f^*_0$, i.e., problem~\eqref{eq:constr1} reduces to \eqref{eqCase0_prob} by letting $C\to \infty$.
		
Suppose 
\beqa
\bs^* \in \arg \min_{\bs \in 
    \mathbb{S}^*_0} w(\bs) \ \mbox{ and } \
	w^* := w(\bs^*) = 
	\min_{\bs \in \mathbb{S}^*_0} w(\bs).
	\label{eq:minbs*}
\eeqa 
Obviously, $w^*$ is the minimum security 
investments necessary to minimize the total
cost in \eqref{eqCase0_prob}, and $\bs^*$ 
is an optimal point with the smallest 
security investments.  
Then, because $f_0$ is nonincreasing, 
for any $C \geq w^*$, we have 
\beqa
f^*_0 
\le f_0(C)  
\leq f_0(w^*) 
\leq w(\bs^*) + \cb\T \pb_{\rm se}(\bs^*)
= f^*_0, 
    \label{eq:mins*1}
\eeqa
which implies $f_0(C) = f_0(w^*) = f^*_0$.
On the other hand, 
\beqa
f^*_0 = f_0(w^*) < f_0(C) \ \mbox{ if } C < w^*,
    \label{eq:mins*2}
\eeqa
where the strict inequality follows from the definition of $\bs^*$ in \eqref{eq:minbs*}; any $\bs$ with $w(\bs) < w^*$ is not an optimal point and, as a result, we have $f^*_0 < w(\bs) + \cb\T \pb_{\rm se}(\bs)$.

The inequalities in \eqref{eq:mins*1}
and \eqref{eq:mins*2} tell us the following: 
increasing the security 
budget $C$ reduces the total cost $f_0(C)$ 
while $C \leq w^*$. On the other hand, beyond $w^*$, 
increasing the budget will not reduce 
the cost any more as $f_0(w^*) \!=\! f^*_0$.

\subsection{Bounds on the Optimal Value of \eqref{eqCase0_prob}} 
    \label{subsec:Case0_main}

From the discussion at the beginning of subsection
\ref{subsec:Prelim}, it is clear that the spectral radius
of the matrix $\diag(\boldsymbol{\alpha}\circ \bs + \boldsymbol{\delta})^{-1}B$ plays an important
role in the dynamics and the determination of a
stable equilibrium 
of \eqref{eq:pdot}. For this reason, we find it
convenient to define the following problem: 
\begin{eqnarray}
\hspace{-0.17in} 
\myb \min_{\bs \ge \0} \big\{ w(\bs)~|~\rho (\diag(\boldsymbol{\alpha}\circ \bs + \boldsymbol{\delta})^{-1}B) \le 1 \big\}
	\label{eq:rho0} 
\end{eqnarray}
Let $C^*$ be the optimal value of this optimization problem, which is the aforementioned minimum investments needed for suppression. 
We show in Appendix~\ref{sec:ProbP0_Balancing} 
that this optimization problem can be 
transformed into a convex (exponential cone) problem and, 
thus, can be solved efficiently.

The following lemma points out an important fact that we 
will make use of in the remainder of the section. 
\begin{lemma} \label{lem:upperC*}
The optimal value $C^*$ of \eqref{eq:rho0} is an upper
bound on $f^*_0$, i.e., $f^*_0 \leq C^*$. 
\end{lemma}
\begin{proof}
\blue{We know that any optimal point $\tilde{\bs}_0$ of \eqref{eq:rho0}
satisfies $w(\tilde{\bs}_0) \!=\! C^*$ and $\pb_{\rm se}(\tilde{\bs}_0) \!=\! \0$. 
Therefore, the total cost achieved by $\tilde{\bs}_0$ is equal to
$C^* + \cb\T \0 \!=\! C^* \!\ge\! f_0(C^*)\!\ge\! f^*_0$.}
\end{proof}

From this lemma and the definition of $\bs^*$ in 
\eqref{eq:minbs*}, we have
\beqa
w^* = w(\bs^*) \leq f^*_0 \leq C^*.
	\label{eq:w*1}
\eeqa  
Obviously, this also implies $f^*_0 = f_0(C)$ for all 
$C \geq C^*$.

In a special case when the equalities 
in \eqref{eq:w*1} hold, we have $f_0^* = w^* = C^*$, 
and an optimal point of the optimization problem 
in \eqref{eq:rho0}, 
say $\tilde{\bs}_0$, is optimal for the problem in 
\eqref{eqCase0_prob} with $\pb_{\rm se}(\tilde{\bs}_0) = 
\0$. 
Also, $\rho\big( \diag(\boldsymbol{\alpha} \circ 
\bs^*+\boldsymbol{\delta})^{-1} B \big) = 1$ and the 
equality holds in the second part of 
Theorem~\ref{thmSpectralRadiusOptimal}.
But, in general these equalities 
may not hold, in which case we must have $w^* < f^*_0 
< C^*$ and $\rho\big( \diag(\boldsymbol{\alpha}\circ \bs^*+\boldsymbol{\delta})^{-1} B \big) > 1$.

Because $w^*$ is unknown beforehand, we cannot determine 
which case holds. However, if we solve the constrained
problem in \eqref{eq:constr1} with $C = C^* - 
\epsilon$ for small positive $\epsilon$, 
either (a) the optimal 
point $\bs^*_C$ we obtain is an optimal point 
of \eqref{eqCase0_prob} if $w(\bs^*_C) < C$ or
(b) $C = C^* - \epsilon \leq w^* \leq f^*_0 
\leq C^*$ if $w(\bs^*_C) 
= C$. Note that in the latter case, we have
$C^* - f^*_0 \leq \epsilon$, and $\tilde{\bs}_0$ is 
$\epsilon$-suboptimal for \eqref{eqCase0_prob}.


At first glance, one may suspect that solving 
\eqref{eq:constr1} is as difficult as solving 
\eqref{eqCase0_prob} because both share the same 
nonconvex objective function. However, \eqref{eq:constr1} 
enjoys a few numerical advantages: first, unlike 
\eqref{eqCase0_prob}, the constrained problem  
\eqref{eq:constr1} with $w(\bs) \le C < C^*$ 
admits only $\pb_{\rm se}(\bs) > \0$, 
which can be computed using the fixed point 
iteration in \eqref{eqFixedPointEquilibrium}. 
Second, perhaps more importantly, because the
stable equilibrium is guaranteed to be strictly
positive, it allows us to employ the approaches 
in Sections~\ref{sec:LowerBound} and 
\ref{sec:UpperBound} in order to compute upper 
and lower bounds on the optimal value $f_0(C)$. 
In the process of finding these bounds, we also 
demonstrate an interesting observation that
we can bound the gap $C^* - f^*_0$ under certain 
conditions.

Before we proceed with discussing the bounds, 
let us remark on the choice of $\epsilon$. 
Theoretically, we want $\epsilon$ as small as 
possible because it determines the 
$\epsilon$-suboptimality of $\tilde{\bs}_0$ when 
the constraint is active (case (b) above). In practice, 
however, $\epsilon$ should not be too small 
because it can cause numerical issues when 
$C^* = w^*$ and $\pb_{\rm se}(\bs^*) = \0$ (or
$\pb_{\rm se}(\bs^*) \approx  \0$ when $C^*\approx 
w^*$). This is because our approaches to finding 
upper and lower bounds on the optimal cost in Sections
\ref{sec:LowerBound} and \ref{sec:UpperBound} rely on 
the strict condition that $\pb_{\rm se}(\bs) > \0$.

\subsubsection{Lower Bound via Convex Relaxation}
    \label{subsubsec:LowerBoundCR}
    
Borrowing a similar approach used in subsection
\ref{subsec:LowerBoundEXPcone}, we can formulate the following convex relaxation of \eqref{eq:constr1}:
\begin{tcolorbox}[colback=white]
\vspace{-0.2in} 
\begin{subequations}
\begin{eqnarray}
\hspace{-0.5in}
{(\rm P_{R3})} ~~
\min_{\bs\geq \0, \pb, \yb\geq \0, U} 
	&&  f(\bs,\pb)= w(\bs)+\cb\T\pb 
	\label{eq:ConvexRelax0} \\
\mathsf{s.t.}&& w(\bs) \le  C \nonumber \\
&& U\1 = B\pb + \boldsymbol{\alpha}\circ \bs 
	+ \boldsymbol{\delta}\nnb\\
&& \pb \ge e^{-\yb}\label{eqCase0RL_p}\\
&& U \ge \diag(e^{\yb})B\diag(e^{-\yb}) 
	\label{eqCase0RL_U}
\end{eqnarray}
\end{subequations}
\vspace{-0.28in}
\end{tcolorbox}
We denote the optimal value of $({\rm P_{R3}})$
by $f_L(C)$. 
The inequalities in \eqref{eqCase0RL_p} and \eqref{eqCase0RL_U} can be expressed as exponential cone constraints as done in \eqref{eqEXP_P_cone} and \eqref{eqEXP_U_cone}, respectively. Thus, this convex relaxation can be solved efficiently.

Clearly, $f_L(C)$ is nonincreasing on $[0, C^*)$.
Also, $f_L(C)$ $\le$ $f_0(C)$ for all $C \in [0,C^*)$. 
Let $f^*_L := \lim_{C\to C^*} f_L(C)$. Together
with the earlier inequality in 
\eqref{eq:w*1}, we have 
\begin{equation}
f_L^* \le f_0^* \le C^*. \label{eqFLstar}
\end{equation}
\blue{More can be said regarding these bounds as follows.}
\begin{thm}\label{thmSolutionOfRelaxed}
\blue{Suppose $C = C^* - \epsilon$ and $(\bs_L,\pb_L,\yb_L,U_L)$ 
is an optimal 
point of $({\rm P_{R3}})$. Let $C_L = w(\bs_L)$. Then, we have
$f_L^* = f_L(C_L) =f_L(C) \le f_0^*$ if $C_L < C$, and }
\begin{equation}
    C^* - f_0^* 
    \leq \epsilon + \epsilon 
	\big( f_L(0) - f_L^* \big)/C^* \ \mbox{ if } 
	C_L = C.  
	\label{eqCstarSubopt}
\end{equation}
\end{thm}
\begin{proof}
Please see Appendix~\ref{ProofThmSolutionOfRelaxed}.  
\end{proof}

\blue{This result} tells us that, when $C_L \!=\! 
C$, any optimal point of \eqref{eq:rho0} is
$O(\epsilon)$-suboptimal for the original problem
in \eqref{eqCase0_prob}. 
\blue{Also, since $f_L(0) \!\le\! f_0(0) \!=\! 
\cb\T\pb_{\rm se}(\0)$, the bound in 
\eqref{eqCstarSubopt} is upper bounded by
$
\epsilon \big( 1
	+ \frac{\cb\T\pb_{\rm se}(\0)}{C^*} \big),
$
which can be computed before solving $({\rm P_{R3}})$.} 
Therefore, a natural 
question that arises is: 
{\em Can we determine if the condition 
$C_L = C$ holds for some $C < C^*$ without 
having to solve the convex relaxation 
$({\rm P_{R3}})$?}

The following theorem offers a (partial) 
answer to this question by providing a 
sufficient condition for the condition 
$C_L = C$ to
hold. For a given budget constraint $C \in \R_+$, 
define $\mathbb{S}_C = \{ \bs \in \RN_+~|~ w(\bs) 
\le C \}.$
		
\begin{thm} \label{claim_suff_cond}
Suppose that every $\bs \in \mathbb{S}_C$ satisfies
\begin{equation}
B\T\diag(\boldsymbol{\alpha})^{-1}\nabla w(\bs) \lneq \cb.  \label{eqClaim_suff_cond}
\end{equation}
Then, 
$C_L= C$. 
If \eqref{eqClaim_suff_cond} holds for 
all $\bs \in \mathbb{S}_{C^*}$, then $f_0^* = C^*$.
\end{thm}
\begin{proof}
A proof can be found in 
Appendix~\ref{appen:thm6}. 
\end{proof}

Note that 
the condition \eqref{eqClaim_suff_cond} can be verified 
prior to solving the relaxed problem $({\rm P_{R3}})$. 
In the case that condition \eqref{eqClaim_suff_cond} 
holds for all $\bs \in \mathbb{S}_{C^*}$, 
any optimal point $\bs_0$ of \eqref{eq:rho0} is also optimal for our original problem in \eqref{eqCase0_prob}.

\subsubsection{Upper Bound via the Reduced Gradient Method}
    \label{subsubsec:UpperBound}
In order to use the RGM for the problem in \eqref{eq:constr1}, 
we first need to introduce a following modification to 
Algorithm~\ref{algReduceGrad}: 
replace line~5 of Algorithm~\ref{algReduceGrad} with
\begin{equation}
\bs^{(t+1)} \gets \cP_{\mathbb{S}_C}\big[ \bs^{(t)} +\gamma_t \big( \boldsymbol{\alpha}\circ\pb^{(t)}\circ \uu\big) \big],
\end{equation}
where $\cP_{\mathbb{S}_C}[\cdot]$ denotes the Euclidean projection onto $\mathbb{S}_C$. 		
When $w$ is simple, this projection step can be very efficient. 
		
The results in Section~\ref{sec:UpperBound} still hold in this case. In particular, at any feasible point $(\bs^\star,\pb^\star)$ of problem $({\rm P})$ with $\boldsymbol{\lambda} = \0$ such that $\bs^\star\in \mathbb{S}_C$, the matrix 
\[
M(\bs^\star) 
= \diag(\boldsymbol{\alpha}\circ\bs^\star 
	+ \boldsymbol{\delta} 
	+ B \pb^\star) 
	- \diag(\1 - \pb^\star) B,
\]
which arises from totally differentiating the constraint $\gb(\sbold, \pb) = \0$, is still a nonsingular M-matrix. 
This can be verified by noting that $M(\bs^\star)$ satisfies $M(\bs^\star)\pb^\star = \pb^\star \circ (B \pb^\star) > \0,$
where the positivity follows from $\pb^\star >\0$ because we require that $\bs^\star\in \mathbb{S}_C$ with $C<C^*$. As a result, we can use Algorithm~\ref{algReduceGrad} with 
an efficient evaluation of reduced gradient as shown in subsection~\ref{subsec:ComputationalComp} and the projection step as described above. \\ \vspace{-0.3in}

\blue{
\begin{rem}
We summarize how to find a good solution to 
\eqref{eqCase0_prob} when $\rho\big( \diag( 
\boldsymbol{\delta})^{-1} B \big) > 1$ based on 
the above discussion:
first, find a pair $(\bs_0, C^*)$ of the optimal 
point and optimal value of \eqref{eq:rho0}. If 
\eqref{eqClaim_suff_cond} holds for all 
$\bs \in \mathbb{S}_{C^*}$, $\bs_0$ is 
an optimal point of \eqref{eqCase0_prob}. 
Otherwise, solve ${\rm (P_{R3})}$
with $C = C^* - \epsilon$ (for small $\epsilon$)
and let $(\bs_L, f_L)$ be the pair of its optimal point 
and optimal value. If $w(\bs_L) = C$, adopt $\bs_0$
as a solution to \eqref{eqCase0_prob} with
\texttt{opt\_{gap}} $\!\leq\!\epsilon ( 1
	\!+\! \frac{\cb\T\pb_{\rm se}(\0)}{C^*})$. 
Otherwise, solve \eqref{eq:constr1} using RGM
and adopt its solution $\bs_U$ as a 
solution to \eqref{eqCase0_prob} with 
\texttt{opt\_{gap}} $\le f_U - f_L$, where 
$f_U = w(\bs_U) + \cb\T \pb_{{\rm se}}(\bs_U)$.
\\ \vspace{-0.2in}
\end{rem}
}










\section{Numerical Results}
    \label{sec:Numerical}

In this section, we provide some numerical results 
that demonstrate the performance of the proposed 
algorithms. Our numerical studies are
carried out in MATLAB (version 9.5) on a laptop with 8GB RAM and 
a 2.4GHz Intel Core i5 processor. 
We consider 5 different strongly 
connected scale-free networks with the power law
parameter for node degrees set to 1.5, 
and the minimum and maximum 
node degrees equal to 2 and $\lceil 3\log N 
\rceil$, respectively, in order to ensure
network connectivity with high probability.
   
For all considered networks, we fix $\alpha_i = 1$ 
and  $\delta_i = 0.1$ for all $i \in \mathcal{A}$. 
The infection rates $\beta_{j,i}$, $(j,i) \in 
\mathcal{E}$, are modeled using i.i.d.
Uniform(0,1) random variables. 
We choose 
\[
w(\bs) = \1\T\bs \quad \mbox{ and } \quad 
\cb = \nu B\T\1 + 2\cb_{\rm rand},
\]
where the elements of 
$\cb_{\rm rand}$ are given by i.i.d. Uniform(0,1) 
random variables and $\nu \ge 0$ is a varying parameter. 
We select $\cb$ above, in 
order to reflect an observation that nodes which 
support more neighbors should, on the average, 
have larger economic costs modeled by $c_i^e$ 
(Section~\ref{sec:Formulation}-A).
We consider two separate cases: 
$\boldsymbol{\lambda} > \0$ 
and $\boldsymbol{\lambda} = \0$. 

\subsection{Case $\boldsymbol{\lambda} > \0$} 
We generate $\boldsymbol{\lambda}$ using i.i.d. 
Uniform(0,1) random variables for each network, 
set $\nu \in \{0, 0.5, 1\}$, and apply 5 schemes 
described below. The results 
\blue{(averaged over 10 runs)}  
are summarized in Table~\ref{C2_tb_comparison_all}, 
and a more detailed description of the simulation setups can be found 
in Appendix~\ref{appen:NumericalSetup}. 
Here, the reported optimality gap is 
the relative optimality gap given by 
${\sf opt\_gap} = (f^{{\rm cur}} - 
{f}^*_{\rm R})/{{f}^*_{\rm R}}$, where ${f}^*_{\rm R} 
:= \min(f^*_{\rm R1}, f^*_{\rm R2})$, 
$f^*_{\rm R1}$ and $f^*_{\rm R2}$ are the optimal 
values of $(\rm{P_{R1}})$ and $(\rm{P_{R2}})$, 
respectively, and $f^{{\rm cur}}$ is the cost
achieved by the solution found by the algorithm 
under consideration. Although
the two convex relaxations are equivalent, their
numerical solutions are not necessarily identical
and we take a conservative lower bound given by
the minimum of the two values.
When the optimal value of a convex relaxation 
is unavailable, we take the other optimal value. 
Also, the column $t_s$ indicates the total runtime.

\begin{table*}[t] 
	\caption{Numerical Results 
	$(\boldsymbol{\lambda} > \0)$.}
	\label{C2_tb_comparison_all}
	\centering
	\begin{adjustbox}{width=0.99\textwidth}%
		\begin{tabular}{|c||c|c|l||c|c|c||c|c|l||c|c|l||c|c|l|}
			\hline
			{$\nu = 0$}
			& \multicolumn{3}{c||}{\tt M-Matrix + OPTI}
			& \multicolumn{3}{c||}{\tt K-Exp + MOSEK}
			& \multicolumn{3}{c||}{\tt RGM + ARMIJO}
			& \multicolumn{3}{c||}{\tt K-Exp SCP}
			& \multicolumn{3}{c|}{\tt M-Matrix SCP} \\ 
			\hline 
			\hline
			\rule[0ex]{0pt}{9pt}
			$N, |\mathcal{E}|$ 
			& ${\sf opt\_gap}$ 	& ${\sf iter}$ &  $t_s$\! (s)
			& ${\sf opt\_gap}$ 	& ${\sf iter}$ &  $t_s$\! (s)
			& ${\sf opt\_gap}$ 	& ${\sf iter}$ &  $t_s$\! (s)
			& ${\sf opt\_gap}$ 	& ${\sf iter}$ &  $t_s$\! (s) 
			& ${\sf opt\_gap}$ 	& ${\sf iter}$ &  $t_s$\! (s) 
			\\ 
			\hline 
			\rule[0ex]{0pt}{8pt}
            $\tt 100, 474$
  &  $\tt 8.89e-2$  &  $\tt 3 , 15$  &  $\tt 0.41$ 	
  &  $\tt 8.89e-2$  &  $\tt 13$  &  $\tt 0.09$ 			
  &  $\tt 1.16e-2$  &  $\tt 12 , 8$  &  $\tt 0.00$ 	
  &  $\tt 1.16e-2$  &  $\tt 4 , 14$  &  $\tt 0.41$ 	
  &  $\tt 1.16e-2$  &  $\tt 4 , 11$   &  $\tt 0.13$ 	
               \\ 
            $\tt 200, 1014$
    &  $\tt 1.13e-1$  &  $\tt 3 , 18$  &  $\tt 1.71$ 	
  &  $\tt 1.13e-1$  &  $\tt 15$  &  $\tt 0.21$ 			
  &  $\tt 1.31e-2$  &  $\tt 8 ,  7$  &  $\tt 0.00$ 	
  &  $\tt 1.31e-2$  &  $\tt 4 , 16$  &  $\tt 0.77$ 	
  &  $\tt 1.31e-2$  &  $\tt 4 , 13$   &  $\tt 0.23$ 	
               \\ 
            $\tt 499, 2738$
    &  $\tt 1.36e-1$  &  $\tt 3 , 27$  &  $\tt 27.8$ 	
  &  $\tt 1.36e-1$  &  $\tt 16$  &  $\tt 0.66$ 			
  &  $\tt 1.26e-2$  &  $\tt 8 ,  8$  &  $\tt 0.01$ 	
  &  $\tt 1.26e-2$  &  $\tt 5 , 18$  &  $\tt 3.53$ 	
  &  $\tt 1.26e-2$  &  $\tt 4 , 13$   &  $\tt 0.63$ 	
               \\ 
            $\tt 999, 5750$
    &  $\tt 1.41e-1$  &  $\tt 4 , 36$  &  $\tt 428$ 	
  &  $\tt 1.41e-1$  &  $\tt 18$  &  $\tt 1.76$ 			
  &  $\tt 1.40e-2$  &  $\tt 11 ,  8$  &  $\tt 0.03$ 	
  &  $\tt 1.40e-2$  &  $\tt 5 , 20$  &  $\tt 8.86$ 	
  &  $\tt 1.40e-2$  &  $\tt 4 , 17$   &  $\tt 2.80$ 	
               \\ 
            $\tt 2001, 12076$
    &\multicolumn{3}{c||}{ \text{n/a}  }	 
  &  $\tt 1.47e-1$  &  $\tt 17$  &  $\tt 4.78$ 			
  &  $\tt 1.37e-2$  &  $\tt 17 , 7$  &  $\tt 0.15$ 	
  &  $\tt 1.37e-2$  &  $\tt 4 , 19$  &  $\tt 16.1$ 	
  &  $\tt 1.37e-2$  &  $\tt 4 , 108$   &  $\tt 4.63$ 	
               \\ 
			\hline			
		\end{tabular}
	\end{adjustbox}

	\vspace{2ex}
	\begin{adjustbox}{width=0.99\textwidth}%
		\begin{tabular}{|c||c|c|l||c|c|c||c|c|l||c|c|l||c|c|l|}
			\hline
			{$\nu = 0.5$}
			& \multicolumn{3}{c||}{\tt M-Matrix + OPTI}
			& \multicolumn{3}{c||}{\tt K-Exp + MOSEK}
			& \multicolumn{3}{c||}{\tt RGM + ARMIJO}
			& \multicolumn{3}{c||}{\tt K-Exp SCP}
			& \multicolumn{3}{c|}{\tt M-Matrix SCP} \\ 
			\hline 
			\hline
			\rule[0ex]{0pt}{9pt}
			$N, |\mathcal{E}|$ 
			& ${\sf opt\_gap}$ 	& ${\sf iter}$ &  $t_s$\! (s)
			& ${\sf opt\_gap}$ 	& ${\sf iter}$ &  $t_s$\! (s)
			& ${\sf opt\_gap}$ 	& ${\sf iter}$ &  $t_s$\! (s)
			& ${\sf opt\_gap}$ 	& ${\sf iter}$ &  $t_s$\! (s) 
			& ${\sf opt\_gap}$ 	& ${\sf iter}$ &  $t_s$\! (s) 
			\\ 
			\hline 
			\rule[0ex]{0pt}{8pt}
            $\tt 100, 474$
    &  $\tt 1.60e-2$  &  $\tt 2 , 20$  &  $\tt 0.33$ 	
  &  $\tt 1.60e-2$  &  $\tt 11$  &  $\tt 0.08$ 			
  &  $\tt 3.24e-3$  &  $\tt 20 , 8$  &  $\tt 0.01$ 	
  &  $\tt 3.24e-3$  &  $\tt 6 , 16$  &  $\tt 0.58$ 	
  &  $\tt 3.24e-3$  &  $\tt 5 , 16$   &  $\tt 0.22$ 	
               \\ 
            $\tt 200, 1014$
    &  $\tt 1.16e-2$  &  $\tt 3 , 20$  &  $\tt 1.88$ 	
  &  $\tt 1.16e-2$  &  $\tt 12$  &  $\tt 0.18$ 			
  &  $\tt 1.98e-3$  &  $\tt 21 , 7$  &  $\tt 0.01$ 	
  &  $\tt 1.98e-3$  &  $\tt 5 , 15$  &  $\tt 0.94$ 	
  &  $\tt 1.98e-3$  &  $\tt 4 , 20$   &  $\tt 0.37$ 	
               \\ 
            $\tt 499, 2738$
    &  $\tt 1.26e-2$  &  $\tt 3 , 29$  &  $\tt 31.4$ 	
  &  $\tt 1.26e-2$  &  $\tt 13$  &  $\tt 0.57$ 			
  &  $\tt 2.26e-3$  &  $\tt 14 , 10$  &  $\tt 0.02$ 	
  &  $\tt 2.26e-3$  &  $\tt 5 , 17$  &  $\tt 3.84$ 	
  &  $\tt 2.26e-3$  &  $\tt 5 , 18$   &  $\tt 1.08$ 	
               \\ 
            $\tt 999, 5750$
    &  $\tt 1.52e-2$  &  $\tt 3 , 29$  &  $\tt 229$ 	
  &  $\tt 1.52e-2$  &  $\tt 15$  &  $\tt 1.60$ 			
  &  $\tt 2.33e-3$  &  $\tt 19 , 11$  &  $\tt 0.06$ 	
  &  $\tt 2.33e-3$  &  $\tt 6 , 17$  &  $\tt 9.66$ 	
  &  $\tt 2.33e-3$  &  $\tt 5 , 24$   &  $\tt 5.50$ 	
               \\ 
            $\tt 2001, 12076$
    &\multicolumn{3}{c||}{ \text{n/a}  }	 
  &  $\tt 1.56e-2$  &  $\tt 15$  &  $\tt 4.53$ 			
  &  $\tt 2.33e-3$  &  $\tt 25 , 8$  &  $\tt 0.22$ 	
  &  $\tt 2.33e-3$  &  $\tt 5 , 18$  &  $\tt 21.3$ 	
  &  $\tt 2.33e-3$  &  $\tt 5 , 208$   &  $\tt 10.7$ 	
               \\  
			\hline			
		\end{tabular}
	\end{adjustbox}

	\vspace{2ex}
	\begin{adjustbox}{width=0.99\textwidth}%
		\begin{tabular}{|c||c|c|l||c|c|c||c|c|l||c|c|l||c|c|l|}
			\hline
			{$\nu = 1$}
			& \multicolumn{3}{c||}{\tt M-Matrix + OPTI}
			& \multicolumn{3}{c||}{\tt K-Exp + MOSEK}
			& \multicolumn{3}{c||}{\tt RGM + ARMIJO}
			& \multicolumn{3}{c||}{\tt K-Exp SCP}
			& \multicolumn{3}{c|}{\tt M-Matrix SCP} \\ 
			\hline 
			\hline
			\rule[0ex]{0pt}{9pt}
			$N, |\mathcal{E}|$ 
			& ${\sf opt\_gap}$ 	& ${\sf iter}$ &  $t_s$\! (s)
			& ${\sf opt\_gap}$ 	& ${\sf iter}$ &  $t_s$\! (s)
			& ${\sf opt\_gap}$ 	& ${\sf iter}$ &  $t_s$\! (s)
			& ${\sf opt\_gap}$ 	& ${\sf iter}$ &  $t_s$\! (s) 
			& ${\sf opt\_gap}$ 	& ${\sf iter}$ &  $t_s$\! (s) 
			\\ 
			\hline 
			\rule[0ex]{0pt}{8pt}
            $\tt 100, 474$
              &  $\tt 3.32e-9$  &  $\tt 1 , 24$  &  $\tt 0.21$ 	
  &  $\tt 2.0e-10$  &  $\tt 11$  &  $\tt 0.08$ 			
  &  $\tt 7.58e-8$  &  $\tt 80 , 9$  &  $\tt 0.02$ 	
  &  $\tt 6.75e-7$  &  $\tt 15 , 12$  &  $\tt 1.20$ 	
  &  $\tt 5.12e-7$  &  $\tt 15 , 18$   &  $\tt 0.84$ 	
               \\ 
            $\tt 200, 1014$
              &  $\tt 2.88e-9$  &  $\tt 1 , 34$  &  $\tt 1.03$ 	
  &  $\tt 8.0e-10$  &  $\tt 13$  &  $\tt 0.20$ 			
  &  $\tt 1.52e-7$  &  $\tt 94 , 8$  &  $\tt 0.05$ 	
  &  $\tt 1.02e-6$  &  $\tt 16 , 14$  &  $\tt 2.92$ 	
  &  $\tt 8.49e-7$  &  $\tt 16 , 21$   &  $\tt 1.59$ 	
               \\ 
            $\tt 499, 2738$
               &  $\tt 2.91e-9$  &  $\tt 1 , 35$  &  $\tt 11.5$ 	
  &  $\tt 1.5e-10$  &  $\tt 11$  &  $\tt 0.51$ 			
  &  $\tt 1.06e-7$  &  $\tt 116 , 10$  &  $\tt 0.16$ 	
  &  $\tt 1.03e-6$  &  $\tt 18 , 16$  &  $\tt 12.8$ 	
  &  $\tt 8.09e-7$  &  $\tt 18 , 24$   &  $\tt 5.68$ 	
               \\ 
            $\tt 999, 5750$
              &  $\tt 2.57e-9$  &  $\tt 2 , 51$  &  $\tt 278$ 	
  &  $\tt 3.61e-9$  &  $\tt 12$  &  $\tt 1.30$ 			
  &  $\tt 1.57e-7$  &  $\tt 126 , 14$  &  $\tt 0.39$ 	
  &  $\tt 1.20e-6$  &  $\tt 18 , 17$  &  $\tt 29.4$ 	
  &  $\tt 9.41e-7$  &  $\tt 18 , 27$   &  $\tt 23.9$ 	
               \\ 
            $\tt 2001, 12076$
              &\multicolumn{3}{c||}{ \text{n/a} }	 
  &  $\tt 1.50e-9$  &  $\tt 12$  &  $\tt 3.76$ 			
  &  $\tt 1.25e-7$  &  $\tt 140 , 9$  &  $\tt 1.35$ 	
  &  $\tt 1.07e-6$  &  $\tt 19 , 17$  &  $\tt 76.7$ 	
  &  $\tt 2.79e-5$  &  $\tt 16 , 663$   &  $\tt 105$ 	
               \\   
			\hline			
		\end{tabular}
	\end{adjustbox}
\end{table*}

$\bullet$ \underline{\tt M-matrix + OPTI}: We solve 
the relaxed problem based on M-matrix in 
subsection~\ref{subsec:LowerBound} using 
Algorithm~\ref{algRelaxation}, where line~3 utilizes an 
interior point method from the OPTI package 
\cite{currie12opti}, and consider the feasible point
$(\tilde{\bs}, \tilde{\pb})$ in 
Theorem~\ref{thmConvexRelaxSolution}. 
The column ${\sf iter}$ shows the pair of (i) the
number of outer 
updates (each corresponding to an approximation 
$\tilde{\Omega}(\underline{\zb}^{(t)})$ of the set 
$\Omega$) and (ii) the average number of inner 
interior-point iterations inside outer updates. 
As we can see, 
the algorithm runtime does not scale 
well with the network size; for the case $(N, 
|\mathcal{E}|) = (2001,12076)$, the solver failed 
to converge within an hour. 
	
$\bullet$ \underline{\tt K-Exp + MOSEK}: We solve 
the relaxed problem $(\rm{P_{R2}})$ with exponential 
cone constraints using the MOSEK package \cite{mosek} 
and consider the feasible point $(\bs', \pb')$ 
in Theorem~\ref{thmConvexRelaxSolution_Cone}.   
The column {\sf iter} indicates the number of interior 
point iterations. As expected, this method enjoys
smaller runtimes and, hence, has
a computational advantage over the 
\mbox{\tt M-matrix + OPTI} scheme. 
	
$\bullet$ \underline{\tt RGM + ARMIJO}: We use the 
RGM in Algorithm~\ref{algReduceGrad} to 
find a local minimizer $(\bs^\star, \pb^\star)$. 
The column {\sf iter} shows the pair of (a) the number 
of gradient updates and 
(b) the maximum number of fixed point 
iterations needed for evaluating $\uu$ in line~4 and 
$\pb^*$ in line~7, denoted by  $\bar{k}_{\rm fp}$.
In our study, the reported values of  
$\bar{k}_{\rm fp}$ are all relatively 
small as expected from our earlier discussions
(Theorem~\ref{thmEquilibrium} and subsection
\ref{subsec:ComputationalComp}).

$\bullet$ \underline{\tt M-Matrix SCP}: We use
Algorithm~\ref{algSCP} to find a local minimizer. 
Here, we solve the convex optimization subproblem 
in line~5 of Algorithm~\ref{algSCP} approximately, 
using the OPTI package for $N \le 10^3$ and a gradient 
descent method for $N > 10^3$. The column 
${\sf iter}$ shows the pair of (a) the number of outer 
linearization updates and (b) the average inner 
steps of either the interior-point solver or gradient 
descent method. However, this approach does not scale well 
due to the M-matrix based relaxation as shown in 
Table~\ref{C2_tb_comparison_all}.

$\bullet$ \underline{\tt K-Exp SCP}: 
For this algorithm, we replace the convex optimization 
subproblem in line~5 of Algorithm~\ref{algSCP} with 
the formulation in \eqref{SCP_ExpC} and solve it 
approximately using MOSEK. The 
column ${\sf iter}$ shows the pair of (a) the number 
of outer linearization updates and (b) the average 
inner interior-point steps. 
As we can see from Table~\ref{C2_tb_comparison_all}, 
this approach achieves similar ${\sf opt\_gap}$ as 
{\tt RGM + ARMIJO}, but the runtime is roughly two 
orders higher. Compared to {\tt M-Matrix SCP},
its performance, both in terms of the quality of
solution and runtime, is comparable.

We summarize observations. First, 
as $\nu$ increases and infection costs become 
larger, as expected from Lemma
\ref{corExactRelaxation}, the gap diminishes and 
becomes negligible when $\nu=1$. 
Second, the upper bounds from local minimizers 
are very close to 
the lower bound $f^*_{\rm R}$, even when 
the relaxation may not be exact (for $\nu =0$ and $0.5$). 
Moreover, they lead to optimal points when the 
relaxation is exact. This suggests that the algorithms 
can practically find global solutions to the original 
problem. Finally, Algorithm~\ref{algReduceGrad} 
based on {\tt RGM}, is highly scalable: 
despite  a larger number of required iterations 
compared to all other schemes, the total runtime~$t_s$ 
is much smaller and is a fraction of that of
Algorithm~\ref{algRelaxation} or~\ref{algSCP}; {we note that we did not optimize step sizes; we instead used the same parameters in all cases.} 
    
\blue{
We also tried {\tt sqp} and {\tt interior-point} solvers in 
{\tt MATLAB} for problem $({\rm P})$, but found them to be
very inefficient 
compared to our approaches to finding local optimizers.
For example, for the case $(N, |\mathcal{E}|) = (999, 5750)$ 
and $\nu=0.5$, while RGM runs in only a fraction of a second, 
{\tt sqp} takes 19 iterations in 102 seconds to achieve 
the same ${\sf opt\_gap}$ as RGM, and {\tt interior-point} 
terminates after 125 iterations in 68 seconds with twice the
${\sf opt\_gap}$ of RGM.}

\subsection{Case $\boldsymbol{\lambda} = \0$} 

In this subsection, we study the scenario with 
no primary attacks, using the 
scale-free network with $499$ nodes from 
the previous subsection. We consider the value of 
$\nu$ in $\{0.6, 0.8, 1\}$ in order to obtain 
more informative numbers. 

\begin{figure*}[h]
	\centering
	\includegraphics[scale=0.62]{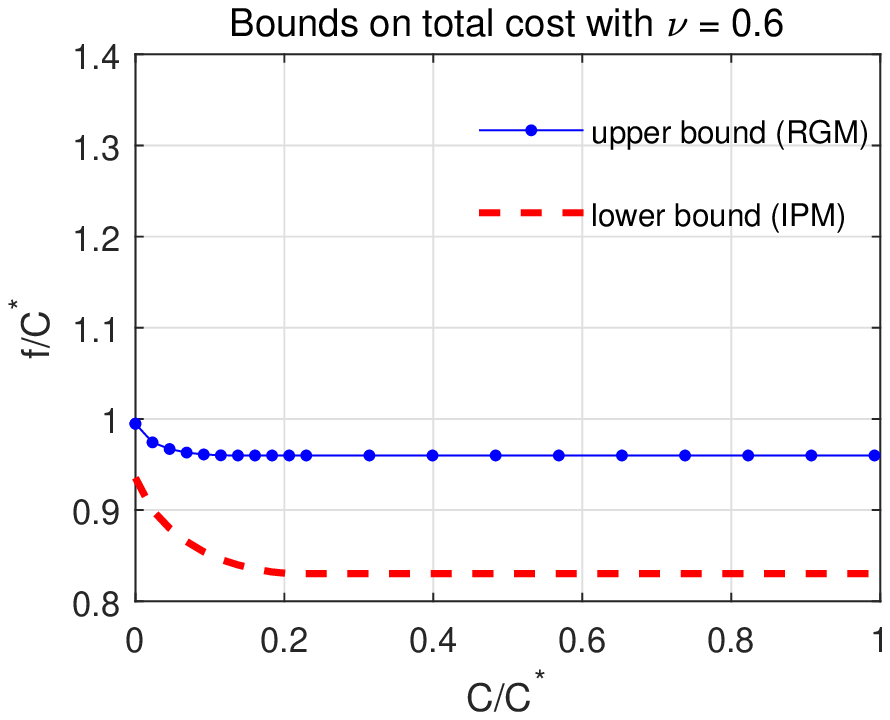}
	\hspace{-6mm}
	\includegraphics[scale=0.62]{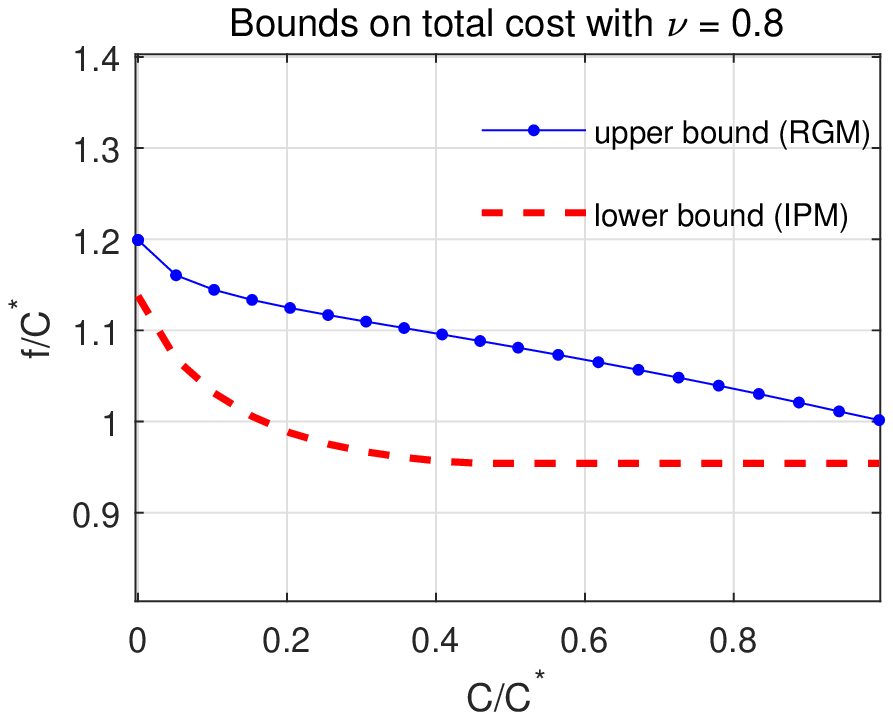}
	\hspace{-6mm}
	\includegraphics[scale=0.62]{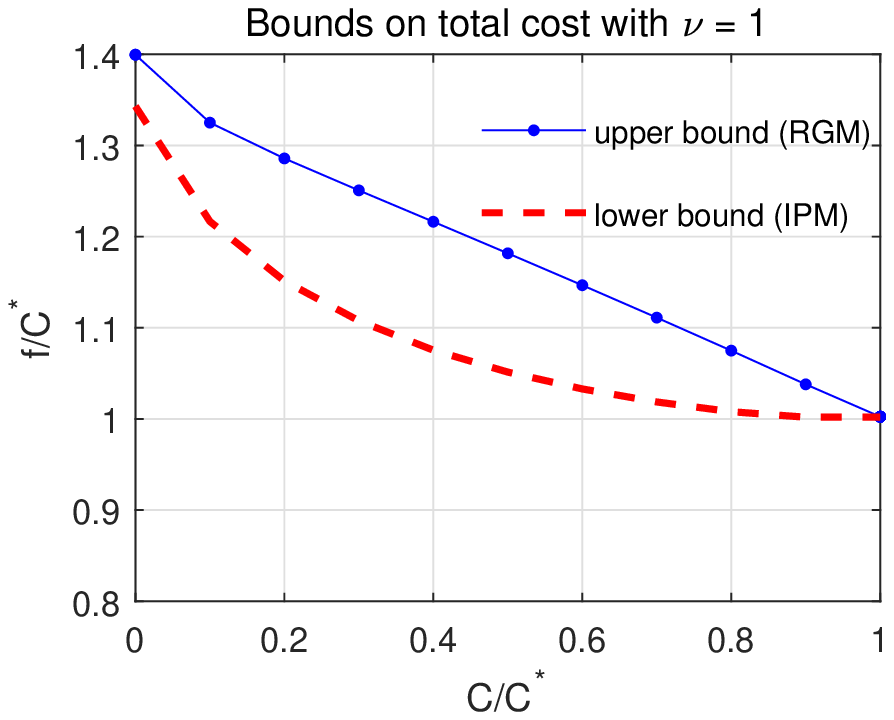}
\caption{Upper and lower bounds on the optimal 
    value of the problem in \eqref{eq:constr1} with 
    varying \emph{fictitious} security budget $C$ 
    ($\nu = 0.6, 0.8$ and $1.0$).}
	\label{fig_lambda0}
\end{figure*}

Following the steps outlined in Section
\ref{sec:Case_lambda0}, 
we first find the optimal value $C^*$ and an optimal 
point $\bs_0$ of problem \eqref{eq:rho0}, using MOSEK. 
When $\nu=1$, condition 
\eqref{eqClaim_suff_cond} in Theorem
\ref{claim_suff_cond} holds for all $\bs 
\in \RN_+$ and, thus, we have $f_0^* = C^*$ 
with $\bs_0$ being an optimal point of the
original problem in \eqref{eqCase0_prob}. 

Second, for $\nu<1$, we consider the problem 
in \eqref{eq:constr1} with $\epsilon = 0.01C^*$
or, equivalently, $C_1 = 0.99C^*$, and find a lower 
bound $f_L(C_1)$, which is the optimal 
value of the relaxed problem $({\rm P_{R3}})$, using 
MOSEK. For $\nu = 0.6$ and $0.8$, we found 
that the constraint $w(\bs) \leq C_1$ is inactive at 
the optimal point and, consequently, 
$f_L^* = f_L(C_1)$. 
In addition, using the projected RGM described in 
subsection \ref{subsubsec:UpperBound}, 
we also compute an upper bound $f_U(C_1)$ 
on the optimal value $f^*_0$ 
and then consider the gap 
$\Delta f_0 := 
\min\{ f_U(C_1), C^* \} - f_L^*$.

We plot in Fig.~\ref{fig_lambda0} both the upper 
bound $f_U(C)$
and the lower bound $f_L(C)$ of the 
optimal value of problem in \eqref{eq:constr1} as
a function of $C$ over the 
interval $[0, 0.99C^*]$. There are several 
observations we make from the plots. 

\noindent {\bf o1.} 
When $\nu = 0.6$ and the infection costs $\cb$ 
are small, the left plot shows $f_U(C) 
< C^*$, which tells us that the  
local minimizer found by the projected RGM
of the problem in \eqref{eq:constr1} 
with the given security budget $C$,
is better than the optimal point of 
\eqref{eq:rho0}. 
The plot also indicates that both the upper and 
lower bounds quickly reach a plateau less than 
$C^*$ with increasing $C$. This likely suggests 
that the optimal security investment at an optimal 
point of \eqref{eqCase0_prob} is significantly 
smaller than $C^*$. 

\noindent {\bf o2.} As we discussed just before 
subsection \ref{subsubsec:LowerBoundCR}, the 
left plot highlights the practical usefulness 
of our approach: by considering the problem in
\eqref{eq:constr1} for $C = C^* - \epsilon$ 
with small positive $\epsilon$, we can quickly 
estimate the optimal security investments, i.e., 
whether or not they are close to $C^*$
(when the security budget constraint is active
at $\bs^*_C$)
or is equal to $w(\bs^*_C)$ (when the constraint
is inactive at $\bs^*_C$), without suffering from
the numerical issues mentioned earlier.

\noindent {\bf o3.}
As $\nu$ increases, so do the bounds 
$f_U(C_1)$ and 
$f_L(C_1)$ (normalized by $C^*$). 
For larger values of $\nu$ with higher 
infection costs (middle and right plots), the
upper bound $f_U(C_1)$ obtained 
from a local minimizer is slightly
larger than another upper bound $C^*$,
suggesting that the budget constraint is active
at a local minimizer returned by the projected
RGM. 
This suggests that, when infection costs are 
high, $\bs_0$, which may overinvest compared to 
an optimal point, may still be a good feasible 
point. Furthermore, as $\nu$ gets close to one, 
eventually $\bs_0$ becomes an optimal point 
for the problem in \eqref{eqCase0_prob}
as shown in the right plot for $\nu = 1$. 
This is expected because as the
infection costs become larger, the system 
operator has an incentive to invest more 
in security.

\noindent {\bf o4.}
Although we do not report detailed numbers here, 
both the upper and lower bounds can be computed 
efficiently using RGM and MOSEK; the 
required computational time is always less 
than 2 seconds for each run, suggesting that
the RGM may be a good method for identifying
suitable security investments for large
systems.

\section{Discussion}
    \label{sec:Discussion}

\subsection{Constraints on Security Investments}
    \label{sec:ConstraintDiscussion}

\blue{
As mentioned in Section~\ref{sec:Formulation}, we 
imposed only non-negativity constraints on 
security investments $\bs$ in problem (P) and
subsequent problems. In this subsection, we 
discuss how additional constraint(s) on $\bs$, 
such as a budget constraint, affect 
our main results reported in Sections~\ref{sec:LowerBound} 
through \ref{sec:Case_lambda0}.
}

\subsubsection{Case with $\boldsymbol{\lambda} \gneq \0$}

\blue{
Suppose that the security investments $\bs$ is 
required to lie in some convex set $\cS \subset \R_+^N$
in problem (P). Then, the relaxed problems 
${\rm (P_{R1})}$ and ${\rm (P_{R2})}$ are still 
convex. However, the pair $(\tilde{\bs}, \tilde{\pb})$
(resp. $(\bs', \pb')$) in Theorem 
\ref{thmConvexRelaxSolution} (resp. 
\ref{thmConvexRelaxSolution_Cone}) is a feasible
point of problem (P) if and only if $\tilde{\bs}
\in \cS$ (resp. $\bs' \in \cS$). 
For this reason, the convex relaxations
${\rm (P_{R1})}$ and ${\rm (P_{R2})}$ are
exact if $\tilde{\bs}$ and $\bs'$ lie in $\cS$
and the condition \eqref{eqExactRelaxationCond2}
in Lemma~\ref{corExactRelaxation} holds. 
In addition, in order to ensure that $\bs^{(t+1)}$ 
belongs to $\cS$, line 6 of Algorithm~\ref{algReduceGrad} 
needs to be modified as follows:
\beqan
\bs^{(t+1)} 
\leftarrow  \cP_{\cS}\big[ \bs^{(t)} 
    - \gamma_t \big( \nabla w(\bs^{(t)}) 
    - \boldsymbol{\alpha}\circ\pb^{(t)}\circ \uu\big) \big], 
\eeqan
where $\cP_{\cS}[\cdot]$ denotes the projection operator
onto $\cS$. 
}

\subsubsection{Case with $\boldsymbol{\lambda} = \0$}

\blue{
The finding in Theorem~\ref{thmSpectralRadiusOptimal}
continues to hold when the minimum element $\bs_{\min}$
of the set $\cS$ exists, with the minimum element
$\bs_{\min}$
being the unique optimal point of the problem 
in \eqref{eqCase0_prob} when 
$\rho\big( \diag(\boldsymbol{\alpha} \circ 
\bs_{\min} + \boldsymbol{\delta})^{-1} B \big) \leq 1$. 
Hence, when the recovery rates $\boldsymbol{\delta}$
are sufficiently large, only the minimum investments 
given by $\bs_{\min}$ are needed. Obviously, when 
$\cS = \R_+^N$, the minimum element is $\0$. Also, 
for a general constraint set ${\cal S}$,
the problem in \eqref{eq:rho0} is not guaranteed to 
be feasible, i.e., $\rho (\diag(\boldsymbol{\alpha}
\circ \bs + \boldsymbol{\delta})^{-1}B) > 1$ for all 
$\bs \in \cS$. This
means that $\pb_{\rm se}(\bs) > \0$ for all $\bs \in 
{\cal S}$ and our methods in Sections~\ref{sec:LowerBound}
and \ref{sec:UpperBound} can be applied directly to the problem 
in \eqref{eqCase0_prob}.
}

\subsection{Relaxation of Irreducibility of $B$}
    \label{subsec:Irreducibility}

\blue{  
Although we suspect that irreducibility of matrix $B$ is a 
reasonable assumption for many systems of interest, 
such as enterprise intranets, some systems may not
satisfy this assumption. For this reason, here we discuss 
how relaxing this assumption affects our results.
}

\blue{
Note that the irreducibility 
of $B$ is used to (i) ensure the existence of a
unique equilibrium $\pb^*(\bs) \in (0,1)^N$ of \eqref{eq:pdot} as shown in Theorem~1, 
and (ii) make use of Lemma~\ref{lemConvexityInvMMatrix} 
for our M-matrix based convex formulations. 
}

\blue{
We relax the assumption that $B$ is irreducible and
instead assume that, for every system $i \in \cA$, either 
$\lambda_i > 0$ or there is another system $j \in \cA
\setminus \{i\}$ with $\lambda_j > 0$ and a directed path 
to $i$ in ${\cal G}$. Then, 
the main results in Theorem~1 still hold, i.e., 
there is a unique equilibrium $\pb^*(\bs) \in (0,1)^N$ of 
\eqref{eq:pdot} which is strictly positive and can be computed 
via iteration~\eqref{eqFixedPointEquilibrium}. 
Moreover, our formulations and results based on exponential 
cones, including the convex relaxation ${\rm (P_{R2})}$ and 
Theorem~\ref{thmConvexRelaxSolution_Cone}, are still 
valid because they rely only on the 
positivity of $\pb(\bs)$. However, the convexity of 
problem ${\rm (P_{R1})}$ is not guaranteed and requires 
extending Lemma~\ref{lemConvexityInvMMatrix}.
Finally, when the aforementioned assumption does not hold, 
the problem is more complicated; our results cannot
be applied directly, and it is still an open problem. 
}



\section{Conclusions}
    \label{sec:Conclusion}
    
We studied the problem of determining suitable 
security investments for hardening interdependent 
component systems of large systems against 
malicious attacks and infections. Our formulation
aims to minimize the average aggregate costs of 
a system operator based on the steady-state analysis. 
We showed that the resulting optimization problem
is nonconvex, and proposed a set of algorithms for
finding a good solution; two approaches are based
on convex relaxations, and the other two look for
a local minimizer based on RGM and SCP. In addition, 
we derived a sufficient condition under which the
convex relaxations are exact. Finally, we evaluated
the proposed algorithms and demonstrated that,
although the original problem is nonconvex, 
local minimizers found by the RGM and SCP methods
are good solutions with only small optimality gaps. 
In addition, as predicted by our analytical results, 
when the infection costs are high, the optimal
points of convex relaxations solve the original
nonconvex problem.

\bibliographystyle{plain}
\bibliography{ref}

\begin{IEEEbiographynophoto}
{Van Sy Mai}  
received his B.E. degree in Electrical Engineering 
from the Hanoi University of Technology in 2008,
his M.E. degree in Electrical Engineering 
from the Chulalongkorn University in 2010, and
his Ph.D. degree in Electrical and Computer Engineering 
from the University of Maryland in 2017. 
Since 2017, he has been a guest researcher 
at the National Institute of Standards and Technology.
\end{IEEEbiographynophoto}

\vspace{-10mm}
\begin{IEEEbiographynophoto}
{Richard J. La} received his Ph.D. degree in Electrical Engineering from the University of California, Berkeley in 2000. Since 2001 he has been on the faculty of the Department of Electrical and Computer Engineering at the University of Maryland, where he is currently a Professor. He is currently an associate editor for IEEE/ACM Transactions on Networking, and served as an associate editor for IEEE Transactions on Information Theory and IEEE Transactions on Mobile Computing. 
\end{IEEEbiographynophoto}

\vspace{-10mm}
\begin{IEEEbiographynophoto}
{Abdella Battou} 
is the Division Chief of the Advanced Network Technologies Division, within The Information Technology Lab at NIST. He also leads the Cloud Computing Program.
His research areas in Information and Communications Technology (ICT)  include cloud computing,  high performance optical networking, information centric networking, and more recently quantum networking.
From 2009 to 2012, prior to joining NIST, he served as the Executive Director of The Mid-Atlantic Crossroads (MAX) GigaPop. From 2000 to 2009, he was Chief Technology Officer, and Vice President of Research and Development for Lambda OpticalSystems.
%
%
Dr. Battou holds a PhD and MSEE in Electrical Engineering, and MA in Mathematics all from the Catholic University of America.
\end{IEEEbiographynophoto}
\newpage


\begin{appendices}
		
\section{Proof of Theorem~\ref{thmEquilibrium}}
    \label{appen:Theorem1}

The existence and uniqueness of an equilibrium 
can be established following the same approach in \cite{khanafer2016stability}, and we omit its
proof here. 

We now show that $\pb^*$ can be obtained using
the iteration in \eqref{eqFixedPointEquilibrium}
of the theorem. Given $\boldsymbol{\lambda} \gneq \0$, 
$\boldsymbol{\delta}, \sbold \ge \0$, it follows that 
the steady state $\pb^*$ satisfies \eqref{eqSteadyState2},
which can be rewritten as 
\begin{align}
(\1-\pb^*)\circ (\boldsymbol{\lambda} + B\pb^*) 
= \uu\circ \pb^*
\label{eqSteadyState3}
\end{align}
with $\uu = \boldsymbol{\alpha}\circ \bs + 
\boldsymbol{\delta}$. This is equivalent to 
\begin{align}
\pb^*  
= \frac{\boldsymbol{\lambda} + B\pb^*}{\boldsymbol{\lambda} 
	+ B\pb^* + \uu} =: H(\pb^*). 
\label{eqSteadyState4}
\end{align}
This suggests that $\pb^*$ is a fixed point of the 
map $H$. It can be verified that $H$ satisfies 
the following properties of a \emph{standard 
interference function} \cite{yates1995framework}: 
for any $\pb \ge  \pb'\ge \0$, (i) 
$H(\pb) > \0$, (ii)  $ H(\pb) \ge  H(\pb')$, 
and (iii) $c H(\pb) >  H(c\pb)$ for $c>1$. As 
a result, the unique fixed point of $H$ can be 
computed using the iteration 
\begin{equation}
\pb_{k+1} = H(\pb_k) 
	\label{eqEqilibriumFixedPoint}
\end{equation}
starting from any $\pb_0 \ge \0$. 
The convergence of the 
synchronous version (as well as asynchronous one) 
can be found in \cite{yates1995framework}. 
However, little is known about the convergence rate 
except for when a stronger condition than (iii) is used. 
For example, linear convergence is achieved when 
$H$ is a contraction \cite{feyzmahdavian2012contractive}. 
Here, we demonstrate linear convergence using a 
different, more direct argument. 

First, for any $\pb \ge  \pb'\ge \pb^*$, we have 
\begin{align*}
\0 \le  H(\pb) -  H(\pb') 
&= \frac{\uu \circ B(\pb - \pb')}
	{(\uu+\boldsymbol{\lambda} + B\pb)\circ
	(\uu+\boldsymbol{\lambda} + B\pb')} \\
&\le \frac{\uu \circ B(\pb - \pb')}
	{(\uu+\boldsymbol{\lambda} + B\pb^*)^2} \\
&= \frac{(\1-\pb^*)^2}{\uu}\circ B(\pb - \pb') \\
&= \diag(1-\pb^*)\Xi B(\pb - \pb'),
\end{align*}
where $\Xi = \diag(\frac{\1-\pb^*}{\uu})$, and
the second equality is a consequence of 
$\1-\pb^* = \frac{\uu}{\uu + 
\boldsymbol{\lambda} + B\pb^*}$ 
from \eqref{eqSteadyState4}. 
Here, for a vector $\bx$, we use $\bx^2 = \bx\circ\bx$. 

Next, we show that $\rho(\Xi B) \le 1$, which then 
implies 
\begin{equation} 
\rho_0:=\rho\big(\diag(\1-\pb^*)\Xi B\big) 
\le 1-\textstyle\min_i  p_i^* 
< 1. \label{eqContractionParam}
\end{equation} 
To this end, note $\rho(\Xi B) \le \rho 
\Big( \Xi \big(\diag(\frac{\boldsymbol{\lambda}}{\pb^*}) 
+ B\big) \Big) = 1$, 
where the inequality follows from 
the monotonicity of the spectral radius of nonnegative 
matrices \cite[Thm~8.1.18]{horn85matrix}, 
and the equality from the fact that 
$\pb^*$ is a strictly positive eigenvector of 
$\Xi \big(\diag(\frac{\boldsymbol{\lambda}}{\pb^*}) 
+ B\big)$ corresponding to eigenvalue $1$, i.e.,  
$$\Xi \Big(\diag(\frac{\boldsymbol{\lambda}}{\pb^*}) 
+ B\Big) \pb^* = \Xi (\boldsymbol{\lambda} +B \pb^*) \stackrel{\eqref{eqSteadyState3}}{=} \pb^*.$$ 
Thus, we conclude that, starting from any $\pb_0$ such that 
$\pb^* \le \pb_0 \le \1$, $H$ is contractive with 
parameter $\rho_0< 1$ given in \eqref{eqContractionParam}. 
In fact, iteration 
\eqref{eqEqilibriumFixedPoint} satisfies $0\le 
\pb_{k} - \pb^* \le \rho_0^k(\pb_0 - \pb^*)$, and
linear convergence is achieved.

\section{A Proof of Convexity of $\Omega$}
	\label{Appendix_CvxOmega}

We will prove this convexity claim by using the 
convexity of the dominant eigenvalue of an 
essentially nonnegative matrix (also known as 
a Metzler matrix, i.e., off-diagonal elements 
are nonnegative). In particular, it is known 
that any Metzler matrix $A\in \Rnn$ has an 
eigenvalue $\bar\sigma(A)$ called the dominant 
eigenvalue that is real and greater than or 
equal to the real part of any other eigenvalue 
of $A$. The following result was shown in 
\cite{cohen1981convexity}. 
\myskip

\noindent
{\em Theorem \cite{cohen1981convexity}:}
Let $A$ be a Metzler matrix and $D$ be a diagonal 
matrix. Then $\bar{\sigma}(A+D)$ is a convex 
function of $D$.
\myskip

\noindent This theorem implies that 
$\underline{\sigma}(\diag(\zb)-B)$ is concave 
in $\zb$. 

Fix $\zb_1, \zb_2\in \Omega$ and let $\zb(\tau) 
:= \tau\zb_1 + (1-\tau)\zb_2$, $\tau \in 
[0, 1]$. We will show that $\zb(\tau) \in \Omega$ 
for all $\tau \in [0,1]$, which proves 
the convexity of $\Omega$. 
Becase $B$ is a Metzler matrix, from
Lemma~\ref{lem_M_Matrix}-(e)), the following 
equivalence holds:
$\zb\in \Omega$ if and only if
$\underline{\sigma}(\diag(\zb)-B) >0$.
Since $\zb_1, \zb_2 \in \Omega$, we have $\underline{\sigma}(\diag(\zb_i)-B) 
>0$ for $i = 1, 2$. This, together with the 
concavity of $\underline{\sigma}(\diag(\zb)-B)$, 
gives us the inequality 
\begin{eqnarray*}
&& \hspace{-0.2in}
\underline{\sigma}(\diag(\zb(\tau))-B) \\ 
\mygeq \tau \underline{\sigma}(\diag(\zb_1)-B) 
    + (1-\tau)\underline{\sigma}(\diag(\zb_2)-B) 
\end{eqnarray*}
which is strictly positive for all $\tau \in [0,1]$.

\section{Formulation of \eqref{MatBal} 
    as a Matrix Balancing Problem}
	\label{Appendix_MatrixBalancing}

First, recall that a matrix $A\in \Rnn_+$ is 
\emph{balanced} if $A\1 = A\T\1$. A positive diagonal 
matrix $X$ is said to balance a matrix $A \in \Rnn_+$ 
if $XAX^{-1}$ is balanced. The following lemma
holds \cite{schneider1990comparative}: 

\begin{lemma} \label{lemMatBalExist}
For any $A \in \Rnn_+$, let $\mathcal{G}_A$ denote 
the weighted directed graph associated with $A$. 
Then, the following statements are equivalent: 
\begin{itemize}
	\item[(i)] A positive diagonal matrix $X$ that 
		balances $A$ exists. 
	\item[(ii)] The graph $\mathcal{G}_A$ is strongly 
		connected. 
	\item[(iii)] There exists a positive vector $\xx$ 
		that minimizes the sum $\sum_{1 \le i,j \le n} 
		a_{ij} x_i x_j^{-1}$.
\end{itemize}
Moreover, matrix $X$ and vector $\xx$ are unique up 
to scalars.%
\end{lemma}

We now proceed to prove that \eqref{MatBal} is a 
matrix balancing problem. Suppose that 
$\underline{\sigma}(\diag(\zb)-B) = 0$, which is 
equivalent to the condition that $\diag(\zb) - B$ is 
a singular M-matrix (i.e.,	$\zb\in \partial\Omega$). 
Then, there exists $\tilde{B} \in \R^{N \times N}_+$ 
such that 
$\diag(\zb) - B = \rho(\tilde{B})I - \tilde{B}.$ 
Since $\tilde{B}$ is also irreducible, it follows 
from the Perron-Frobenius theorem \cite{horn85matrix} 
that there is $\xx > \0$ such that $\tilde{B}\xx = 
\rho(\tilde{B}) \xx$. As a result, $(\diag(\zb) - B)
\xx = \0$, i.e., 
\begin{align}
\textstyle\sum_{j=1}^N b_{ij}x_jx_i^{-1} = z_i, 
	\quad \forall i=1,\ldots,n. \label{eqEigvec}
\end{align}
Therefore, we can write \eqref{MatBal} as follows.
\begin{align}
&\textstyle\min_{\zb > \0} \quad\{\hb\T\zb ~|~ 
	\underline{\sigma} \big(\diag(\zb) - B\big) 
	= 0 \} \nnb\\
= &\textstyle\min_{\zb > \0, \xx >\0} \quad\{\hb\T\zb 
	~|~ \eqref{eqEigvec} \text{ holds}\} \nnb\\
= &\textstyle\min_{\xx >\0} \quad  \sum_{1\le i,j\le n} 
	h_ib_{ij}x_jx_i^{-1} \nnb
\end{align}
which, by Lemma~\ref{lemMatBalExist}, is the problem 
of balancing $\diag(\hh)B$.

\section{A Proof of Theorem 
\ref{thmConvexRelaxSolution_Cone}}
    \label{appen_exact_relax_cone}

The first inequality in \eqref{eq:thm3} is obvious 
because $({\rm P_{R2}})$ is a relaxation of 
$({\rm P})$. The second inequality follows from 
the feasibility of  $(\bs', \pb')$ as shown below. 

\subsection{Proof of Feasibility of $(\bs', \pb')$
    for (P)}

Recall that 
\[
\pb' = e^{-\yb^+} > \0 \quad 
	\bs' = \bs^+ + \diag(\boldsymbol{\alpha}^{-1})
	    B(\pb^+ - \pb').
\]
Since $\bs^+ \ge \0$ and $\pb^+ \ge e^{-\yb^+} = 
\pb'$, it follows that $\bs' \ge \0$. It remains 
to show that $(\bs', \pb')$ satisfies the constraint 
\eqref{eqSteadyState} or, equivalently, 
\begin{equation}
(\pb')^{-1}\circ\boldsymbol{\lambda} 
    + (\pb')^{-1}\circ B\pb' 
= \boldsymbol{\lambda} + B\pb' 
    + \boldsymbol{\alpha}\circ \bs' 
    + \boldsymbol{\delta}.
    \tag{\ref{eqSteadyState2}}
\end{equation}
From the definition of $(\pb',\bs')$, 
\eqref{eqSteadyState2} is equivalent to 
\begin{align}
e^{\yb^+}\circ\boldsymbol{\lambda} + e^{\yb^+}\circ Be^{-\yb^+} 
= \boldsymbol{\lambda} + B\pb^+ + \boldsymbol{\alpha}\circ \bs^+ + \boldsymbol{\delta},\nnb
\end{align}
where the right-hand side equals $\tb^+ + U^+\1$
from \eqref{eqSteadyState6}. 

To show that the above equality holds, we will prove 
that the equalities
\begin{equation} 
\tb^+ = \blambda \circ e^{\yb^+}, \quad 
U^+ = \diag(e^{\yb^+})B\diag(e^{-\yb^+}) 
    \tag{\ref{eqEXP_varopt}}
\end{equation}
hold, i.e., the inequality constraints of $\tb$ and 
$U$ are active at the solution $\bx^+_{\rm R}$ of 
problem $({\rm P_{R2}})$. 
Let $\bmu_s$, $\bmu_p, \overline{\bmu}_p$, $\bmu_t$, 
$\bmu_y \in \RN_+$ and $\Phi \in \mathbb{R}^{N\times N}_+$ 
be the Lagrangian multipliers associated with inequality 
constraints, and $\boldsymbol{\sigma} \in \RN$ with 
the equality constraint in \eqref{eqSteadyState6} of 
$({\rm P_{R2}})$. Then, the Lagrangian function is 
given by 
\begin{align}
L = &w(\bs) + \langle\cb,\pb\rangle 
+ \langle \bmu_p, e^{-\yb}-\pb \rangle  
+ \langle \bmu_t,\blambda\circ e^{\yb} 
    - \tb \rangle \nnb\\
&+ \langle \Phi, \diag(e^{\yb})B\diag(e^{-\yb}) 
    - U\rangle \nnb\\
&- \langle \bmu_s, \bs \rangle 
    + \langle \overline{\bmu}_p, \pb-\1 \rangle
    - \langle \bmu_y,\yb \rangle\nnb\\
&+ \langle \bsigma, \tb + U\1 
    - \boldsymbol{\lambda} - B\pb 
    - \boldsymbol{\alpha}\circ \bs - \boldsymbol{\delta} \rangle. \nnb
\end{align}
Since the problem is convex, the following 
Karush-Kuhn-Tucker (KKT) conditions are 
necessary and sufficient for optimality.
\begin{subequations}
\begin{eqnarray}
\0 \myeq \nabla_p L 
    = \cb - \bmu_p + \overline{\bmu}_p - B\T\bsigma \label{L_p}\\
\0 \myeq \nabla_t L 
    = \bsigma - \bmu_t \label{L_t}\\
0 \myeq \partial_{u_{ij}} L 
    = \sigma_i - \phi_{ij} \ \mbox{ for all }
        (i,j) \in \mathcal{E} 
    \label{L_u}\\
0 \myeq \partial_{y_i} L 
    = \mu_{t_i} \lambda_i e^{y_i^+} 
    - \mu_{p_i}e^{-y_i^+} -\mu_{y_i}
    + \sum_{j \in \mathcal{A}}\phi_{ij}b_{ij} 
        e^{y_i^+ \!-y_j^+} \nnb\\
&& \qquad
    - \sum_{j \in \mathcal{A}} \phi_{ji}b_{ji} 
        e^{y^+_j\!-y^+_i} \ \mbox{ for all } 
        i \in \mathcal{A} 
        \label{L_y}\\
\0 \myeq \bmu_t \circ (e^{\yb^+}-\tb^+) \label{slk_t}\\
0 \myeq \phi_{ij}(u^+_{ij} - b_{ij}e^{y^+_i-y^+_j})
    \ \mbox{ for all } (i,j) \in \mathcal{E}. 
    \label{slk_u}
\end{eqnarray}
\end{subequations}
From \eqref{L_t} and \eqref{L_u}, we get
\begin{equation}
\sigma_i = \mu_{t_i} = \phi_{ij} \ge 0 
    \ \mbox{ for all } (i,j) \in \mathcal{E} .
    \label{L_tu}
\end{equation}
Combining this and \eqref{L_y} yields 
\begin{eqnarray*}
&& \hspace{-0.2in}
\sigma_i (\lambda_i e^{y_i^+}          
    + \sum_{j}b_{ij}e^{y_i^+-y_j^+}) \\
\myeq \mu_{p_i}e^{-y_i^+} \!+\! \mu_{y_i} 
    + \sum_{j}\sigma_j b_{ji}e^{y_j^+ - y_i^+} \\
\mygeq (\mu_{p_i} + \sum_{j}\sigma_i b_{ji}) 
    e^{-y_i^+}
\geq (c_i+\overline{\mu}_{p_i})e^{-y_i^+} 
> 0,
\end{eqnarray*}
where the first inequality follows from $\bmu_y\ge \0$ and $\yb^+\ge \0$, and the second from~\eqref{L_p}. Thus,  $\bsigma > \0$, which together with \eqref{L_tu} and the slackness conditions \eqref{slk_t} and \eqref{slk_u} implies that \eqref{eqEXP_varopt} must hold.

\subsection{Equivalence of Relaxed Problems 
$({\rm P_{R1}})$ and $({\rm P_{R2}})$}

We will show that $({\rm P_{R1}})$ and $({\rm P_{R2}})$ are equivalent. Recall
\begin{subequations}
\beqa
\hspace{-0.4in}
(\mathrm{P_{R1}})\qquad  
\min_{\bs, \pb, \zb}
	&&  f(\bs, \pb) \lb
\mathsf{s.t.}&&\pb \ge  
	(\diag(\zb) - B)^{-1}\boldsymbol{\lambda} 
	\nnb\\
&& \zb = \boldsymbol{\alpha} \circ \bs 
	+ \boldsymbol{\delta} 
	+ \boldsymbol{\lambda} + B\pb \nnb\\
&& \bs \ge \0,\quad \pb \le \1,
	\quad \zb \in \Omega.
	\nonumber
\eeqa
\end{subequations}
Let $\vv := (\diag(\zb) - B)^{-1}\boldsymbol{\lambda}$. Then, we have 
\begin{equation}
(\diag(\zb) - B)\vv = \boldsymbol{\lambda}. 
    \label{eqP1_vv}
\end{equation}

Since $B$ is irreducible, Lemma~\ref{lem_M_Matrix}-(f) 
implies that if $\vv > \0$, we have $\zb 
\in \Omega$ because $\boldsymbol{\lambda} \gneq \0$. 
In addition, 
we can show that 
if $\zb \in \Omega$, then $(\diag(\zb) - B)^{-1} > 0$
and thus $\vv > \0$.\footnote{We can show that the inverse is 
strictly positive as follows: $\zb \in \Omega$ implies that $\diag(\zb) - B \!=\! \mu(I - A)$ for some 
$\mu \!>\!0$ and $A \!\in\! \R^{N \times N}_+$, which is irreducible 
with positive diagonal elements and satisfies $\rho(A) \!<\! 1$. Thus, 
$(\diag(\zb) - B)^{-1} 
\!=\! \mu^{-1}\!\sum_{k\ge 0}A^k$. Lemma 8.5.5 in 
\cite{horn85matrix} tells us that irreducibility implies $A^{N-1} \!>\! 0$, 
hence $\sum_{k \geq 0}\! A^k \!>\! 0$.
}
Thus, $\zb \in \Omega$ if and only if $\vv > \0$
and, when $\vv > \0$, \eqref{eqP1_vv} is equivalent 
to
\begin{align}
    \diag(\vv^{-1})(\diag(\zb) - B)\vv = \diag(\vv^{-1})\boldsymbol{\lambda}.\nnb
\end{align}
Rearranging the terms, we get
\begin{align}
\zb 
= \vv^{-1}\circ B\vv + \vv^{-1}\circ \boldsymbol{\lambda}.
    \label{eqP1_vv1}
\end{align}
Substituting the expressions for $\vv$ and $\zb$
in the first and second constraint, respectively, 
and replacing the constraint $\zb \in \Omega$
with $\vv > \0$,
we can reformulate $(\mathrm{P_{R1}})$ as follows.
\beqa
\hspace{-0.4in}
\min_{\bs, \pb, \vv}
	&&  f(\bs, \pb) 
	\label{eq_PR1_vv}\\
\mathsf{s.t.}&&\pb \ge  \vv  \nnb\\
&& \vv^{-1}\circ B\vv + \vv^{-1}\circ \boldsymbol{\lambda} = \boldsymbol{\alpha} \circ \bs 
	+ \boldsymbol{\delta} 
	+ \boldsymbol{\lambda} + B\pb\nnb\\
&& \bs \ge \0,\quad \pb \le \1,
	\quad \vv > \0.
	\nonumber
\eeqa

Next, we will show that this problem is equivalent to 
$(\mathrm{P_{R2}})$. Recall that, as we showed in the 
previous subsection, at any optimal point $\bx^+ 
= (\bs^+, \pb^+, \yb^+, \tb^+, U^+)$ of 
$(\mathrm{P_{R2}})$, equalities in \eqref{eqEXP_varopt} 
hold. Thus, $(\mathrm{P_{R2}})$ can be reduced to the 
following problem after eliminating $\tb$ and $U$. 
\begin{align}
\hspace{-0.5in}
\min_{\bs, \pb, \yb}&
\quad f(\bs, \pb) \label{eq_PR2_reduced}\\
\mathsf{s.t.} &\quad \pb \ge e^{-\yb} \nnb\\
&\quad \blambda\circ e^{\yb} 
    + e^{\yb} \circ B e^{-\yb} 
= \boldsymbol{\alpha} \circ \bs 
	\!+\! \boldsymbol{\delta} 	\!+\! \boldsymbol{\lambda} \!+\! B\pb\nnb\\
&\quad \bs \ge \0,\quad \pb \le \1,\quad \yb \ge \0. \nnb
\end{align}
Clearly, \eqref{eq_PR2_reduced} is identical 
to \eqref{eq_PR1_vv} after replacing $e^{-\yb}$
with $\vv$. This proves the equivalence of 
$({\rm P_{R1}})$ and $({\rm P_{R2}})$.

\section{A Proof of Theorem \ref{thmSCP_Cone_Exact}}
    \label{appen_exact_SCP_cone}

In order to prove the theorem, it suffices to 
show that at any optimal point 
$\bx^+= (\bs^+, \yb^+, \tb^+, U^+)$ of the
relaxed problem in \eqref{SCP_ExpC},  
the constraints \eqref{SCP_ExpC_t}
and \eqref{SCP_ExpC_U} are active. In 
other words, 
\begin{align} 
&\tb^+ = \blambda \circ e^{\yb^+}, \quad 
U^+ = \diag(e^{\yb^+})B\diag(e^{-\yb^+}).
    \label{SCP_ExpC_tU}
\end{align}
The proof is similar to that in 
Appendix~\ref{appen_exact_relax_cone}. 
Let $\bmu_s$, $\bmu_t\in \RN_+$ and $\Phi \in 
\mathbb{R}^{N\times N}_+$ be the Lagrange multipliers 
associated with inequality constraints
\eqref{SCP_ExpC_t} and \eqref{SCP_ExpC_U}, and 
$\boldsymbol{\sigma} \in \RN$ those associated with 
the equality constraint \eqref{SCP_ExpC_Linear} of 
the relaxed problem. The Lagrangian is
given by 
\begin{align}
L &=   w(\bs) + \langle\cb,e^{-\yb}\rangle - \langle \bmu_s, \bs \rangle 
+ \langle \bmu_t,\blambda\circ e^{\yb} - \tb \rangle \nnb\\
&\quad + \langle \Phi, \diag(e^{\yb})B^{(t)}\diag(e^{-\yb}) - U\rangle \nnb\\
&\quad + \langle \bsigma, \tb + U\1 
- \boldsymbol{\lambda} - B\pb^{(t)} - \boldsymbol{\alpha}\circ \bs - \boldsymbol{\delta} \rangle .
\end{align}
Since the problem is convex, the necessary and 
sufficient KKT conditions are given by the following:
\begin{subequations}
\begin{eqnarray}
\0 \myeq \nabla_t L 
    = \bsigma - \bmu_t 
    \label{SCP_L_t}\\
0 \myeq \partial_{u_{ij}} L 
    = \sigma_i - \phi_{ij} 
    \ \mbox{ for all } (i,j) \in \mathcal{E}    
    \label{SCP_L_u}\\
0 \myeq \partial_{y_i} L 
    = -c_ie^{-y_i^+} + \mu_{t_i} \lambda_i e^{y_i^+} 
    + \sum_{j} \phi_{ij} b^{(t)}_{ij}
    e^{y_i^+ - y_j^+}\nnb\\
&&\qquad
    - \sum_{j} \phi_{ji}b^{(t)}_{ji} e^{y^+_j - y^+_i}
    \ \mbox{ for all }  i \in \mathcal{A} 
    \label{SCP_L_y}\\
\0 \myeq \bmu_t \circ (e^{\yb^+}-\tb^+) 
    \label{SCP_slk_t}\\
0 \myeq \phi_{ij}(u^+_{ij} - b^{(t)}_{ij}e^{y^+_i-y^+_j}) 
    \ \mbox{ for all } (i,j) \in \mathcal{E}. 
    \label{SCP_slk_u}
\end{eqnarray}
\end{subequations}
From \eqref{SCP_L_t} and \eqref{SCP_L_u},
\begin{equation}
\sigma_i = \mu_{t_i} = \phi_{ij} \ge 0 
    \ \mbox{ for all } (i,j) \in \mathcal{E}.  
    \label{SCP_L_tu}
\end{equation}
By combining this and \eqref{SCP_L_y}, we obtain 
\begin{equation*}
\sigma_i \big[\lambda_i e^{y_i^+} 
    + \sum_{j} b^{(t)}_{ij} e^{y_i^+-y_j^+}\big] 
= c_i e^{-y_i^+} 
    + \sum_{j} \sigma_j b^{(t)}_{ji} e^{y_j^+ - y_i^+},
\end{equation*}
which is strictly positive because $c_i>0$. Thus,  
$\bsigma > \0$, which together with \eqref{SCP_L_tu} 
and the slackness conditions in \eqref{SCP_slk_t} and 
\eqref{SCP_slk_u}, implies that \eqref{SCP_ExpC_tU} 
must hold.

\section{\blue{Proof of Theorem~\ref{thmSolutionOfRelaxed}} } \label{ProofThmSolutionOfRelaxed}
Suppose $(\bs_L,\pb_L,\yb_L,U_L)$ is an optimal 
point of $({\rm P_{R3}})$, and let $C_L = w(\bs_L)$. 
There are two possibilities: $C_L = C$ or $C_L < C$.
We shall consider them separately below.

$\bullet$ {$C_L= C$:} 
In this case, we have
\begin{equation}
C^* - \epsilon = w(\bs_L) < f_L(C),  \label{eqFLC_lowerbound}
\end{equation}	
where the inequality follows from the assumption 
$C < C^*$ and, hence, $\pb_{\rm se}(\bs_L) > \0$. 
Since the relaxed problem  $({\rm P_{R3}}$) 
is convex, $f_L(C)$ is a convex 
function of $C$~\cite{boyd2004convex}. 
Thus, we have
\[
f_L(C) 
\leq \big( 1-\frac{ \epsilon}{C^*} \big) f_L^* 
	+ \frac{\epsilon}{C^*}f_L(0) 
\stackrel{\eqref{eqFLstar}}{\leq}  
f_0^* + \frac{\epsilon}{C^*} \big( f_L(0) 
	- f_L^* \big).
\]
Combining this with \eqref{eqFLC_lowerbound} yields
\begin{equation}
0 \stackrel{\eqref{eq:w*1}}{\leq} 
C^* - f_0^* 
\leq \epsilon + \frac{\epsilon}{C^*}
	\big( f_L(0) - f_L^* \big).  
	\label{eqCstarSubopt1}
\end{equation}	
Suppose that $\bs_0$ is an optimal point of 
\eqref{eq:rho0}. Then, \eqref{eqCstarSubopt1} 
implies that $\bs_0$ is $O(\epsilon)$-suboptimal 
for the original problem in \eqref{eqCase0_prob}. 

$\bullet$ {$ C_L< C$:} Since the constraint
$w(\bs_L) \leq C$ is inactive, it follows that 
$f_L(C_L) = f_L(C')$ for all $C' \in [C_L, C]$.
Because $C$ can be made arbitrarily close to 
$C^*$, this also tells us that $f_L(C_L) = f^*_L$, 
providing us with a following lower bound:
$f_L^* = f_L(C_L) \le f_0^*$.
This lower bound, together with an upper bound
we can obtain as explained in subsection
\ref{subsubsec:UpperBound}, can be used 
to quantify how close to optimal a feasible 
solution is.

\section{A Proof of Theorem~\ref{claim_suff_cond}}
    \label{appen:thm6}

We prove the first part by contradiction: 
suppose that there exists some $C < C^*$ such that, at an 
optimal point $(\bs_L,\pb_L,\yb_L,U_L)$, we have $C_L < C$. 
Consider 
\[ 
\yb' = \yb_L + \gamma\1, \ \ 
\pb'= e^{-\gamma}\pb_L, \ \  
\bs' = \bs_L + (1-e^{-\gamma})\boldsymbol{\alpha}^{-1}
    \circ B\pb_L
\]
for some $\gamma > 0$ such that $w(\bs') < C$. We can find
such positive $\gamma$ because $w$ is continuous and 
increasing, and $w(\bs_L) = C_L<C$. We can verify that  
the tuple $(\bs',\pb',\yb',U_L)$ is a feasible point of
the relaxed problem $({\rm P_{R3}})$. Moreover, 
\begin{align}
&f(\bs',\pb') - f(\bs_L,\pb_L) 
= w(\bs') - w(\bs_L) + \cb\T(\pb'-\pb_L) \nnb\\
&\le \nabla w(\bs')\T(\bs'- \bs_L) + (e^{-\gamma}-1)\cb\T\pb_L \nnb\\
&= (1-e^{-\gamma})\big(B\T \diag(\boldsymbol{\alpha}^{-1})
    \nabla w(\bs') - \cb\big)\T\pb_L < 0.\nnb
\end{align}
The second inequality holds because $\gamma >0$, 
$\pb_L > \0, \bs' \in \mathbb{S}_C$ and the condition \eqref{eqClaim_suff_cond} is true. 
But, this contradicts the assumed optimality of $(\bs_L,\pb_L,\yb_L,U_L)$. Hence, we have $C_L = C$. 
	
To prove the second part of theorem, 
assume that \eqref{eqClaim_suff_cond} holds for 
all $\bs \in \mathbb{S}_{C^*}$. Take any sequence of positive 
$\epsilon_k$, $k \in \N$, which satisfies $\lim_{k \to \infty}
\epsilon_k = 0$, and consider the sequence $\{C_k = 
C^*-\epsilon_k : k \in \N\}$. Because $\mathbb{S}_{C_k} 
\subset \mathbb{S}_{C^*}$, from \eqref{eqCstarSubopt}, 
we obtain
\beqan
C^* - f^*_0 
\leq \epsilon_k + \frac{\epsilon_k}{C^*} 
    \big( f_L(0) - f^*_L \big) 
    \ \mbox{ for all } k \in \N. 
\eeqan
Since the right-hand side goes to 0 as $k \to \infty$,
$f_0^* = C^*$ and $\bs_0$ is an optimal point 
of the original problem in \eqref{eqCase0_prob}.

\section{Reformulation of Problem \eqref{eq:rho0}
    As a Convex (Exponential Cone) Problem} \label{sec:ProbP0_Balancing}
Recall that the optimization problem is given 
by 
\begin{eqnarray*}
(\mathrm{P}_C) \hspace{0.3in}
\min_{\bs\in \mathbb{R}^N_+} &&  w(\bs) \\
\mathsf{s.t.} && \rho (\diag(\boldsymbol{\alpha}\circ \bs + \boldsymbol{\delta})^{-1}B) \le 1. 
\end{eqnarray*}
By Perron-Frobenius theory \cite{horn85matrix}, the above spectral radius constraint is equivalent to 
the existence of a positive vector $\bx$ such 
that $\diag(\boldsymbol{\alpha}\circ \bs + \boldsymbol{\delta})^{-1}B\bx \le \bx$.
With a change of variable $\yb = \log (\xx)$, problem $(\mathrm{P}_C)$ can be expressed as follows.
\begin{align}
\min_{\bs \ge \0,\by}&\quad   w(\bs) \label{eqObjectiveFun_02}\\
\mathsf{s.t.}&\quad \textstyle \sum_{j\in \mathcal{A}} b_{ij}e^{y_j-y_i} \le \alpha_is_i + \delta_i, \quad
\forall i \in \mathcal{A}.
\end{align}
This is a convex problem and can be turned into an exponential cone optimization problem as follows 
with $U = [u_{ij}]\in \mathbb{R}^{N\times N}$: 
\begin{eqnarray}
\min_{\bs\ge \0,\by,U} && \myb  w(\bs)
    \label{eq:ExpCone}\\
\mathsf{s.t.} && \myb  
    U\1 \le \boldsymbol{\alpha}\circ\bs 
    + \boldsymbol{\delta} \nonumber \\
&& \myb (u_{ij}, 1, y_i - y_j +\log b_{ij})
    \in \mathcal{K}_{\rm exp}, \quad
    \forall (i,j) \in \mathcal{E} 
    \nnb
\end{eqnarray}
This problem can be solved efficiently using convex 
solvers \cite{mosek,currie12opti}, provided that $N$ 
is not large.  
			
Below, we describe a first-order algorithm 
for solving the problem in \eqref{eq:ExpCone} 
based on our previous matrix balancing algorithm. This algorithm can deal with large-scale problems in a parallel and distributed fashion.  
For our discussion, without loss of generality, 
we assume $\boldsymbol{\alpha} = \1$. 
Let us consider the Lagrangian given by 
\begin{align}
L(\mathbf{y}, \mathbf{s}, \btheta) 
&= w( \bs) + \sum_{i\in \mathcal{V}} \theta_i \Big( \sum_{j\in \mathcal{V}} b_{i,j} e^{y_j-y_i} - s_i -\delta_i\Big) \nnb\\
&= w(\bs) - \btheta\T(\bs + \boldsymbol{\delta}) + \sum_{(i,j)\in \mathcal{E}}  \theta_i b_{i,j} e^{y_j-y_i},
\end{align}
where $\btheta \ge 0$ is a Lagrange multiplier vector. 
The dual function can be obtained from the
Lagrangian by solving the following problem:

\begin{align*}
g(\btheta) 
&:= \inf_{\mathbf{y}\in \mathbb{R}^N, \mathbf{s}\ge \0} 
    L(\mathbf{y}, \mathbf{s}, \btheta) \nnb\\
&= \inf_{\mathbf{s}\ge \0} w(\bs) - \btheta\T (\bs + \boldsymbol{\delta}) + \inf_{\mathbf{y}} \sum_{(i,j)\in \mathcal{E}}  \theta_i b_{i,j} e^{y_j-y_i}, 
\end{align*}
It is obvious that this problem can be decoupled into 
two subproblems -- one over $\bs$ and the other
over $\yb$ -- as follows.
\begin{align*}
\bs(\boldsymbol{\theta}) 
&= \arg\inf_{\mathbf{s}\ge \0} w(\bs) - \boldsymbol{\theta}\T \bs\\
\yb(\boldsymbol{\theta}) &= \arg\inf_{\mathbf{y}} \sum_{i, j\in \mathcal{E}}  
    \theta_i b_{i,j} e^{y_j-y_i}
\end{align*}
Note that the second problem of finding 
$\yb(\boldsymbol{\theta})$ is a problem of balancing
the matrix $\diag(\boldsymbol{\theta})B$ and can be solved 
efficiently. 

The dual function is now equal to 
\[
g(\btheta)
=  w(\bs(\btheta)) \!-\! \btheta\T (\bs(\btheta) + \boldsymbol{\delta}) 
\!+\!\! \sum_{(i,j)\in \mathcal{E}}  \theta_i b_{ij} e^{y_j(\btheta)-y_i(\btheta)},
\]
and the dual problem given by
\begin{align*}
\max_{\btheta \ge \0} ~~~ g(\btheta) 
\end{align*}
is a convex problem. 
Suppose that the Slater's condition holds for problem 	\eqref{eqObjectiveFun_02} so that we have strong duality. 
Then, we can solve \eqref{eqObjectiveFun_02} by solving 
the dual problem instead, for example, using a 
first-order method with subgradients. 


\section{A Proof of the continuity of Stable
Equilibrium $\pb_{\rm se}$}
    \label{appen:Continuity}
    
We will prove the claim by contradiction. 
Suppose that $\pb_{\rm se}$ is not continuous
at some point $\bs' \in \R^N_+$. 
Then, there exists a sequence
$\bs_k$, $k \in \N$, such that 
$\bs_k \to \bs'$ as $k \to \infty$, but either
(a) $\lim_{k \to \infty} \pb_{\rm se}(\bs_k)$ does 
not exist or (b) the limit exists but is not equal 
to $\pb_{\rm se}(\bs')$. When the limit 
does not exist, we can work with a convergent
subsequence instead; the existence of a convergent
subsequence of $\pb_{\rm se}(\bs_k)$, 
$k \in \N$, is guaranteed by the 
Bolzano-Weierstrass Theorem~\cite{Royden-RealAnalysis}. 
For this reason, without loss of generality, we 
assume $\lim_{k \to \infty} \pb_{\rm se}(\bs_k)$ 
exists and denote it by $\pb^\dagger$. 

We consider the following two cases separately. 
Note that $\pb_{\rm se}(\bs)$ cannot have
both positive and zero entries. 

{\bf Case 1: $\pb_{\rm se}(\bs') > {\bf 0}$ --}
Recall from subsection~\ref{subsec:Model} that 
when $\boldsymbol{\lambda} = \0$, 
\beqan
\gb(\bs, \pb)
= ({\bf 1} - \pb) \circ B \pb
    - ( \boldsymbol{\alpha} \circ \bs
    + \boldsymbol{\delta}) \circ \pb .
\eeqan
Since $\gb(\bs', \pb_{\rm se}(\bs'))
= {\bf 0}$, the implicit function 
theorem \cite[Theorem 3]{Halkin74-SC} 
tells us that, for any set $\mathbb{P}'$
containing $\pb_{\rm se}(\bs')$, we can 
find a set $\mathbb{S}'$ 
containing $\bs'$ and a function 
$\psi:\mathbb{S}' \to \mathbb{P}'$
such that (a) $\psi(\bs') = \pb_{\rm se}(\bs')$, 
(b) $\gb(\bs, \psi(\bs)) = {\bf 0}$ 
for all $\bs \in \mathbb{S}'$, and 
(c) $\psi$ is continuous at $\bs'$. Since
there is at most one positive stable equilibrium
satisfying $\gb(\bs, \pb) = {\bf 0}$, this implies
$\psi(\bs) = \pb_{\rm se}(\bs) > {\bf 0}$ for all 
$\bs$ in a sufficiently small neighborhood around
$\bs'$ and, hence, 
$\pb^\dagger = \pb_{\rm se}(\bs')$, which 
is a contradiction.

{\bf Case 2: $\pb_{\rm se}(\bs') = {\bf 0}$ --}
In this case, since there is a unique solution 
to $\gb(\bs', \pb) = {\bf 0}$, 
$\pb^\dagger$ must be equal to ${\bf 0}$ from 
the continuity of $\gb$, 
which contradicts the earlier assumption 
that $\pb^\dagger \neq \pb_{\rm se}(\bs') = {\bf 0}$.

\section{A Description of the Setup for Numerical
    Studies}
\label{appen:NumericalSetup}

The setup for running our algorithms is as follows. 

\noindent $\bullet$ \underline{M-matrix + OPTI :}
In line 3 of Algorithm 1 for solving 
$(\rm{P_{R1}})$ with $\Omega$ replaced by
$\tilde{\Omega}(\zb^{(t)})$, 
we use a general interior point optimizer  
from package \cite{currie12opti} 
with relative convergence tolerance set to $10^{-5}$ 
and Hessian matrices approximated by a quasi-Newton algorithm. 
The solver initial point 
$\tilde{\bx}_{\rm{R}}^{(0)} 
\!=\! (\tilde{\bs}_{\rm{R}}^{(0)}, 
\tilde{\pb}_{\rm{R}}^{(0)}, \tilde{\zb}_{\rm{R}}^{(0)})$
is chosen to be $\tilde{\bs}_{\rm{R}}^{(0)} \!=\! \0$, 
$\tilde{\pb}_{\rm{R}}^{(0)} \!=\! \pb^*(\0)$ using the
iteration in \eqref{eqFixedPointEquilibrium}, and 
$\tilde{\zb}_{\rm{R}}^{(0)} = \boldsymbol{\lambda}
+ \boldsymbol{\delta} + B \tilde{\pb}_{\rm{R}}^{(0)}$
according to \eqref{eqLinear_zsp}. When computing
$\tilde{\pb}_{\rm{R}}^{(0)}$, the iteration
in \eqref{eqFixedPointEquilibrium} is run until either 
$\lnorm \pb_{k+1} - \pb_{k} \rnorm /
\lnorm \pb_k \rnorm \leq 10^{-7}$ or $k=500$. In addition, 
we approximate the set of active constraints
of $\tilde{\zb}_{\rm{R}}$ (line 4 of Algorithm 1)
using $\mathcal{I}_{ac} = \{i \in 
\mathcal{A} \ | \ \left[\tilde{\zb}_{\rm{R}} \right]_i
- [\zb]_i \leq 10^{-3} \}$, and select $\bar{h}
= 10$.  It is important to note that 
the OPTI package \cite{currie12opti} does not exploit/support multithreading; 
this is also one of the disadvantages compared to MOSEK.

\vspace{1ex}
\noindent $\bullet$ \underline{Exp-cone + MOSEK :} 
We set the relative gap tolerance of the the interior 
point optimizer in MOSEK to be $10^{-5}$ and do not 
use MOSEK's presolve procedure as it is time-consuming 
for large networks. Moreover, by default, the 
interior-point optimizer in MOSEK is parallelized and 
automatically exploits a maximum number of threads.

\vspace{1ex}
\noindent $\bullet$ \underline{RGM :} 
For Algorithm~\ref{algReduceGrad}, we select $(\bs^{(0)}, 
\pb^{(0)}) = (\0, \pb^*(\0))$ as a feasible initial 
point\footnote{If $(\tilde{\bs}, \tilde{\pb}) \!=\! (\tilde{\bs}(\bx^*_{\rm{R}}), 
	\tilde{\pb}(\bx^*_{\rm{R}}))$ given in Theorem~\ref{thmConvexRelaxSolution} is available, it 
	can be used as an initial point. Here, we choose $(\0, 
	\pb^*(\0))$ to be the initial point for numerical 
	comparisons.}  
and stop the algorithm whenever
$
\frac{\|\bs^{(t+1)} - \bs^{(t)}\| }
{ \| \bs^{(t)} \| }
\le 10^{-6}
\mbox{ or }
\frac{|F(\bs^{(t+1)}) - F(\bs^{(t)})| }
{ F(\bs^{(t)}) }
\le 10^{-8}.
$
We compute $\ub$ in line~4 using the fixed point  
iteration in \eqref{eqFixedPointGradient} and 
$\pb^*(\bs^{(t+1)})$ in line~7 using the iteration in \eqref{eqFixedPointEquilibrium}, with stopping
conditions
$
\frac{\|\ub_{k+1} - \ub_{k} \|}{\|\ub_{k}\|} 
\le 10^{-7}
\mbox{ and } 
\frac{\|\pb_{k+1} - \pb_{k} \|}{\|\pb_{k}\|} 
\le 10^{-7}.
$
We employ the backtracking line search algorithm  with $\gamma_0 = 0.5$, a shrinking factor
of 0.85, and the Armijo condition parameter set to $10^{-4}$.

\vspace{1ex}
\noindent $\bullet$ \underline{K-Exp SCP :}
We note that for conic optimization problems, MOSEK 
currently offers only an interior-point type optimizer 
that cannot take advantage of a previous optimal solution 
\cite{mosek}, requiring a cold-start at each outer 
iteration. The relative gap tolerance of the interior 
point optimizer in MOSEK is set to be $10^{-5}$, and 
the presolve 
procedure was ignored. The relative tolerance errors for 
the outer updates are such that $\frac{\|\bs^{(t+1)} 
- \bs^{(t)}\| }
{ \| \bs^{(t)} \| }
\le 10^{-6}
\mbox{ or }
\frac{|F(\bs^{(t+1)}) - F(\bs^{(t)})| }
{ F(\bs^{(t)}) }
\le 10^{-6}.$

\vspace{1ex}
\noindent $\bullet$ \underline{M-matrix SCP :}
We use Algorithm~\ref{algSCP} in combination with the
OPTI package. We approximate $\Omega^{(t)}$ by 
$\Omega^{(t)}_{\rm mb} = \big\{\bs \in \Rn_+| 
\bs \ge \max\{\underline{\bs}, \boldsymbol{\delta}+ \boldsymbol\lambda +  B\pb^{(t)}\} \big\}$, where 
\[
\underline{\bs} = \arg\min_{\bs \in \Rn_+} \big\{ \1\T\sbold ~|~\underline{\sigma}\big(\diag(\boldsymbol{\alpha}\circ\bs) - \diag(1-\pb^{(t)})B \big) = 0 \big\}
\]
which can be converted into a matrix balancing problem. 
Here, we set the relative gap tolerance of the interior 
point optimizer in OPTI to be $10^{-5}$ and the outer 
tolerance to be the same as in \mbox{\tt K-Exp SCP}. 
We use a warm-start at each outer iteration, except for 
the first one. 
    
\vspace{1ex}
\noindent
\blue{$\bullet$ \underline{MATLAB {\tt fmincon}}:
We attempted to use the built-in {\tt fmincon} function in 
MATLAB to solve $({\rm P})$ directly in both variables $(\bs, \bp)$, but found it inefficient compared to the following form:  
$$\min_{\pb\in [0,1]^N}~	\{ w(\bs(\pb)) + \cb\T\pb~|~\bs(\pb) \ge \0 \},$$
where $\bs(\pb) \!=\! (\pb^{-1}\!-\!\1)\circ (\boldsymbol{\lambda} \!+\! B\pb) \!-\! \boldsymbol{\delta}$,  obtained from \eqref{eqSteadyState}. We used {\tt sqp} and {\tt interior-point} algorithms to solve this problem with constraint and 
optimality tolerances set to 
$10^{-5}$.
}
\end{appendices}

\end{document}